\ifx\synctex\undefined\else\synctex=1\fi
\documentclass{llncs}
\usepackage{tikz}
\usepackage{listings}
\usepackage{todonotes}
\newcommand{\igw}[1]{}{}
\newcommand{\igwin}[1]{}{}
\newcommand{\anca}[2]{}{}

\RequirePackage{amsmath}
\RequirePackage{amsfonts}
\RequirePackage{amssymb}
\RequirePackage{theorem}
\RequirePackage{ifthen}
 \newcommand{\Nat}{\ensuremath{\mathbb{N}}}

\newcommand{\es}{\emptyset}

\newcommand{\incl}{\subseteq}
\newcommand{\fleq}{\preccurlyeq}

\newcommand{\sq}[1]{[#1]}
\newcommand{\di}[1]{\langle #1 \rangle}
\newcommand{\sqq}[1]{\sq{\cdot }}
\newcommand{\ddi}[1]{\di{\cdot }}
\newcommand{\set}[1]{\{#1\}}

\renewcommand\bar[1]{\overline{#1}}

\newcommand{\act}[1]{\stackrel{#1}{\longrightarrow}}

\renewcommand{\a}{\alpha}
\renewcommand{\b}{\beta}
\renewcommand{\d}{\delta}
\newcommand{\e}{\varepsilon}

\newcommand{\g}{\gamma}
\renewcommand{\l}{\lambda}

\renewcommand{\t}{\tau}

\newcommand{\D}{\Delta}

\newcommand{\G}{\Gamma}

\renewcommand{\S}{\Sigma}

\newcommand{\Aa}{\mathcal{A}}

\newcommand{\Cc}{\mathcal{C}}

\newcommand{\Nn}{\mathcal{N}}

\newcommand{\Pp}{\mathcal{P}}

\newcommand{\Ss}{\mathcal{S}}

\newcommand{\Vv}{\mathcal{V}}

\newcommand{\LFP}{\operatorname{\mathsf{LFP}}}
\newcommand{\GFP}{\operatorname{\mathsf{GFP}}}

\newcommand{\PSPACE}{\text{\sc Pspace}}
\newcommand{\PTIME}{\text{\sc Ptime}}
\newcommand{\NP}{\text{\sc NP}}

\newcommand{\EXPTIME}{\text{\sc Exptime}}

\newcommand{\struct}[1]{\langle #1 \rangle}
\newcommand{\comment}[1]{\marginpar{\footnotesize #1}}

\def\sqr#1#2{\vbox
 {\hrule height#2
  \mbox{\vrule width#2 height#1 \kern#1 \vrule width#2}%
  \hrule height#2}}

\newcommand{\wh}[1]{\widehat{#1}}

\newcommand\sem[1]{{[\![ #1 ]\!]}}

{\end{list}}

{\end{list}}

{\end{list}}

\newcommand{\ignore}[1]{}
\newcommand{\ai}{\mathit{i}}
\newcommand{\ao}{\mathit{o}}
\newcommand{\ar}{\mathit{r}}
\newcommand{\aw}{\mathit{w}}
\newcommand{\abr}{\bar{\mathit{r}}}
\newcommand{\abw}{\bar{\mathit{w}}}
\newcommand{\agr}{\wh{\iota}}
\newcommand{\agw}{\wh{\mathit{o}}}

\newcommand{\platt}{\mathsf{flat}}

\newcommand{\sig}{\mathsf{sig}}

\newcommand{\CDsystem}{$(C,D)$-system}
\newcommand{\CDsystems}{$(C,D)$-systems}

\newcommand{\init}{\mathsf{init}}

\renewcommand{\gets}{\coloneqq}
\DeclareMathAlphabet{\kw}{\encodingdefault}{\sfdefault}{bx}{n}

\renewcommand{\max}{\textsf{max}}
\newcommand{\spawn}{{\sf spawn}}

\newcommand{\dminus}{\mathbin{\dot{\smash-}}}
\newcommand{\sact}[1]{\stackrel{#1}{\Longrightarrow}}
\newcommand{\lact}[1]{\stackrel{#1}{\dashrightarrow}_L}
\newcommand{\lzact}[1]{\stackrel{#1}{\dashrightarrow}}
\newcommand{\lpact}[1]{\stackrel{#1}{\dashrightarrow}_{L'}}
\newcommand{\lppact}[1]{\stackrel{#1}{\dashrightarrow}_{L''}}

\newcommand{\lIact}[1]{\stackrel{#1}{\dashrightarrow}_{L_i}}
\newcommand{\liact}[1]{\stackrel{#1}{\dashrightarrow}_{L_1}}
\newcommand{\liiact}[1]{\stackrel{#1}{\dashrightarrow}_{L_2}}

\newcommand{\ext}{\mathit{ext}}
\newcommand{\leqtree}{\sqsubseteq}
\newcommand{\sset}{\textsf{set}}
\DeclareMathOperator{\core}{\textsf{core}}

\newcommand{\pref}{\mathsf{pref}}

\newcommand{\mtrees}{\text{\textit{M-trees}}}
\newcommand{\strees}{\text{\textit{S-trees}}}
\newcommand{\Sext}{\Sigma_{\mathit{ext}}}
\newcommand{\Ssp}{\Sigma_{\mathit{sp}}}

\newcommand{\dist}{\mathsf{filter}}
\newcommand{\lift}{\mathsf{lift}}

\newcommand{\Ext}{\mathit{Ext}}

\newcommand{\ssp}{\mathit{sp}}

\newcommand{\DEXPTIME}{\text{\sc DExptime}}

\newcommand{\Consistent}{\mathit{Consistent}}
\newcommand{\lloop}{\!\circlearrowleft\!}
\newcommand{\out}{\mathit{out}}

\begin{document}
 \title{Reachability for dynamic parametric processes}
 \author{Anca Muscholl\inst{1} \and Helmut Seidl\inst{2} \and Igor Walukiewicz\inst{3}}
 \institute{LaBRI, Univ.~Bordeaux and TUM-IAS \and
 	Fakult\"at f\"ur Informatik,
 	TU M\"unchen \and LaBRI, CNRS, Univ.~Bordeaux
 	}
\maketitle

\begin{abstract}
In a dynamic parametric process every subprocess may spawn
arbitrarily many, identical child processes, that may communicate either over
global variables, or over local variables that are shared with their parent.
We show that reachability for dynamic parametric processes is decidable under mild assumptions.
These assumptions are e.g.~met if individual processes are realized by
pushdown systems, or even higher-order pushdown systems.
We also
provide algorithms for subclasses of pushdown dynamic parametric
processes, with complexity ranging between NP and DEXPTIME.
\end{abstract}

\section{Introduction}\label{s:intro}
Programming languages such as {\sf Java, Erlang, Scala} offer the
possibility to generate recursively new threads (or processes, actors,\dots). Threads 
may  exchange data through globally accessible data structures,
e.g.~via static attributes of  classes like in {\sf Java,
  Scala}. In addition,
newly created threads may locally communicate with their parent threads, 
in {\sf Java}, e.g., via the corresponding thread objects, or via
messages like in {\sf Erlang}.


Various attempts have been made to analyze systems with recursion and dynamic creation of threads that may or may not
exchange data.
A single thread executing a possibly recursive program operating on finitely many local data, can
conveniently be modeled by a \emph{pushdown system}. Intuitively, the pushdown formalizes the call stack of the
program while the finite set of states allows to formalize the current program state together with the current
values of the local variables. For such systems reachability of a bad state or a
regular set of bad configurations is decidable \cite{DBLP:conf/cav/Walukiewicz96,DBLP:conf/concur/BouajjaniEM97}.
The situation becomes more intricate if multiple threads are allowed. 
Already for two pushdown threads reachability is undecidable if communication via a 2-bit global is allowed.
In absence of global variables, reachability becomes undecidable already for two pushdown threads if a rendez-vous primitive is
available \cite{DBLP:journals/toplas/Ramalingam00}. A similar result
holds if finitely many locks are allowed 
\cite{DBLP:conf/cav/KahlonIG05}.
Interestingly, decidability is retained if locking is performed in a disciplined way. 
This is, e.g., the case for nested
\cite{DBLP:conf/cav/KahlonIG05} and contextual locking \cite{DBLP:conf/tacas/ChadhaMV12}.
These decidability results have been extended to 
dynamic pushdown networks as introduced by Bouajjani et al. 
\cite{BouajjaniMOT05}.
This model combines pushdown threads with dynamic thread creation by means of a 
\textsf{spawn} operation, while it ignores any exchange of data between threads.
Indeed, reachability of dedicated states or even regular sets of configurations stays decidable
in this model, if finitely many global locks together with nested locking 
\cite{DBLP:conf/sas/LammichM08,DBLP:conf/cav/LammichMW09} or
contextual locking \cite{DBLP:conf/sas/LammichMSW13} are allowed.
Such regular sets allow, e.g., to describe
undesirable situations such as concurrent execution of conflicting operations.
%
%
%

Here, we follow another line of research where models of
multi-threading are sought which allow exchange of data via shared
variables while still being decidable.  The general idea goes back to
Kahlon, who observed that various verification problems become
decidable for multi-pushdown systems that are \emph{parametric}
\cite{DBLP:conf/lics/Kahlon08}, i.e., systems consisting of an
arbitrary number of indistinguishable pushdown threads.  Later, Hague
extended this result by showing that an extra designated leader thread
can be added without sacrificing decidability
\cite{DBLP:conf/fsttcs/Hague11}. All threads communicate here over a
shared, bounded register \emph{without} locking.
It is crucial for decidability that only
one thread has an identity, and that the operations on
the shared variable do not allow to elect a second leader.
Later, Esparza et al.~clarified the complexity of deciding reachability in that model 
\cite{DBLP:journals/jacm/EsparzaGM16}. La Torre et al. generalized these results to 
hierarchically nested models~\cite{DBLP:conf/concur/TorreMW15}.
Still, the question whether reachability is decidable for \emph{dynamically evolving} parametric pushdown processes, remained open. 

We show that reachability is decidable  for a very general class
of dynamic processes with parametric spawn. We require some
very basic properties from the class of transitions systems that
underlies the model, like e.g.~effective non-emptiness check.
\igw{changed description}In our model every sub-process can maintain e.g.~a pushdown store, or even a
higher-order pushdown store, and can communicate over global variables, as
well as via local variables with its sub-processes and with its parent.
As  in
\cite{DBLP:conf/fsttcs/Hague11,DBLP:journals/jacm/EsparzaGM16,DBLP:conf/concur/TorreMW15},
all variables have bounded domains and no locks are allowed.

%
Since the algorithm is rather expensive, we also present meaningful instances where reachability
can be decided by simpler means. As one such instance we consider the situation where communication
between sub-processes is through global variables only. We show that reachability for this model with pushdowns can effectively be reduced 
to reachability in the parametric model of Hague
\cite{DBLP:conf/fsttcs/Hague11,DBLP:journals/jacm/EsparzaGM16}, called
\CDsystems\ ---
giving us a precise characterization of the complexity as \PSPACE. 
As another instance, we consider a parametric variant of \emph{generalized futures} where spawned sub-processes may not only 
return a single result but create a stream of answers.
For that model, we obtain complexities between \NP\ and \DEXPTIME.
This opens the venue to apply e.g.~SAT-solving to check safety
properties of such programs.

\medskip
\noindent
\emph{Overview.} Section~2 provides basic definitions, and the
semantics of our model. In Section~3 we show a simpler semantics,
that is equivalent w.r.t.~reachability. Section~4 introduces some
prerequisites for Section~5, which is the core of the
proof of our main result. Section~6 considers the complexity for some
special instances of dynamic parametric pushdown processes.



\section{Basic definitions}\label{sec:basic-definitions}\label{s:prelim}

In this section we introduce our model of dynamic parametric processes. 
\igw{simplified}We refrain
from using some particular program syntax;  instead we use potentially
infinite state transition systems with actions on transitions. 
Actions may manipulate local or global variables, or
spawn  \emph{parametrically} some sub-processes: this means that an unspecified number of sub-processes is 
created --- all with the same designated initial state. 
Making the  spawn operation parametric is the main abstraction step that
allows us to obtain decidability results.

Before giving formal definitions we present two examples in order to
give an intuitive understanding of the kind of processes we are
interested in.

\begin{example}\label{e:simple}
A parametric system could, e.g., be defined by 
an explicitly given finite transition system:  
\begin{center}
\begin{tikzpicture}[scale=1.5, bend angle=30]
  \node[draw,circle] (i) at (0,0) {$q$};
  \node[draw,circle] (ii) at (1.2,0) {$q_1$};
  \node[draw,circle] (iii) at (2.4,0) {$q_2$};
  \node[draw,circle] (iv) at (3.6,0) {$q_3$};
  \node[draw,circle] (v) at (4.8,0) {$q_4$};
  \node[draw,circle] (vi) at (6,0) {$q_5$};
  \draw[->] (i) -- (ii) node[above,midway]{\small $\spawn(p)$};
  \draw[->] (ii)--(iii) node[above,midway]{\small $\aw(x,1)$};
  \draw[->] (iii)--(iv) node[above,midway]{\small $\ar(x,2)$};
  \draw[->] (iv)--(v) node[above,midway]{\small $\ar(x,3)$};
  \draw[->] (v)--(vi) node[above,midway]{\small $\aw(g_0,\#)$};
%
  \node[draw,circle] (p) at (0,-1) {$p$};
  \node[draw,circle] (pi) at (1.2,-1) {$p_1$};
  \node[draw,circle] (pii) at (2.4,-1) {$p_2$};
  \draw[->] (p) -- (pi) node[above,midway]{\small $\ai(x,1)$};
  \draw[->] (pi) to [out=20,in=160] node[above,midway]{\small $\ao(x,2)$} (pii);
    \draw[->] (pi) to [out=340,in=200] node[below,midway]{\small
      $\ao(x,3)$} (pii);
\draw[->] (pii) to [out=240,in=330] node[above,midway]{\small
  $\tau$} (p);
\end{tikzpicture}
\end{center}

In this example, the root starts in state $q$ by spawning a
number of sub-processes, each starting in state $p$. Then the root
writes the value 1 into the local variable $x$, and waits for some
child to change the value of $x$ first to $2$, and subsequently to $3$.
Only then, the root will write value $\#$ into the global
variable $g_0$.
Every child on the other hand, when starting execution at state $p$, waits for value $1$
in the variable $x$ of the parent and then chooses either to
write $2$ or $3$ into $x$, then returns to the initial state. The read/write operations of the children
are denoted as input/output operations $\ai(x,v), \ao(x,v)$, because
they act on the parent's local.
Note that at least two children are required to write $\#$.
\end{example}

\noindent
More interesting examples require more program states. 
Here, it is convenient to adopt a programming-like notation as in the next example.
\begin{example}\label{ex:prelim}
\begin{figure}[htb]
\begin{lstlisting}
root() {        
   spawn(p);  
   switch (x) {
   case 2:    write(#); 
   }
}

p() {
   switch (parent.x) {
   case 0 :   spawn(p);
              if (*) parent.x = 1 
              else switch (x) {
              case 1 : parent.x = 1; break;
              case 2 : break;
              }; break;
   case 1 :   spawn(p);
              if (*) parent.x = 0
              else switch (x) {
              case 1: parent.x = 2; break;
              case 2: parent.x = 2; break;
               }; 
   }
}
\end{lstlisting}
\caption{\label{f:example}A program defining a dynamic parametric process.}
\end{figure}
Consider the program from Figure~\ref{f:example}. The states of the system correspond to the lines in 
the listing, and $\textsf{if}( *)$ denotes non-deterministic choice.
There is a single global variable which is written to by the call \texttt{write(\#)}, 
and a single local variable $x$ per sub-process, with initial value 0.
The corresponding local of the parent is accessed via the keyword \textsf{parent}.

The question is whether the root can eventually write \texttt\#?
This would be the case if the value of the root's local
variable becomes \texttt{2}. 
This in turn may occur once the variable \texttt{x} of
some descendant is set to \texttt{1}. In order to achieve this,
cooperation of several sub-processes is needed.
Here is one possible execution.
\begin{enumerate}
\item The root spawns two sub-processes in state $p$, say $T_1$ and $T_2$. 
\item $T_1$ changes the value of the local variable of the root to \texttt{1}
  (line 11).
\item  $T_2$ then can take the \texttt{case 1} branch and
first  spawn $T_3$. 
\item $T_3$ takes the \texttt{case 0} branch, spawns
  a new process and changes the value of \texttt{parent.x} to \texttt{1}.
\item As the variable \texttt{parent.x} of $T_3$ is the local variable
  of $T_2$,
  the latter can now take the second branch of the nondeterministic choice and
  change \texttt{parent.x} to \texttt{2} (line 19) --- which is the
  local variable  of the root.
\end{enumerate}
\qed
\end{example}


In the following sections we present a formal definition of our
parametric model, state the reachability problem, and the main results.
This is done in three steps. 
In the first subsection, we introduce the syntax that will be given
in a form of a  transition system, as the one from the first example.
Next, we give the formal operational semantics that
captures the behavior described in the above examples. 
Finally, we formulate general requirements on a class of systems and
state the result saying that for every class satisfying these
requirements, the  reachability problem for the associated dynamic parametric
processes is decidable.

\subsection{Transition systems}
A \emph{dynamic parametric process} $\Ss$ is a transition
system over a dedicated set of action names.
One can think of it as a control flow graph of a program.
In this transition system the action names are uninterpreted. In
Section~\ref{s:multiset} we will define their semantics.
Such a transition system can be obtained by symbolically executing a
program, say, expanding while loops and  procedure calls. 
Another possibility is that the control flow of a program is given
by a pushdown automaton; in this case the transition system
will have configurations of the pushdown automaton as states.

The transition system is specified by a 
tuple $\Ss=\struct{Q,G,X,V,\Delta,q_\init,v_\init}$ consisting of:
\begin{itemize}
\item	a (possibly infinite) set $Q$ of states,
\item	finite sets $G$ and $X$ of global and local variables,
  respectively, and a finite set $V$ of values for variables; these
  are used to define the set of labels,
\item	an initial state $q_\init\in Q$, and  an initial value
  $v_\init\in V$ for variables,
\item	a set of rules $\Delta$ of the form $q\act{a} q'$, 
where the label $a$ is one of the following:
\begin{itemize}
\item $\tau$, that will be later interpreted as a silent action,
\item $\ar(x,v)$, $\aw(x,v)$, will be interpreted as a  read or a
  write of value $v\in V$ from or to 
  a local or global variable $x\in X\cup G$ of the process,
\item $\ai(x,v)$, $\ao(x,v)$, will be interpreted as a read or a write
  of value $v\in V$ to or from 
  a local variable $x\in X$ of the parent process,
\item $\spawn(q)$, will be interpreted as a spawn of an arbitrary
  number (possibly zero) of new sub-processes, all 
  starting in state $q\in Q$. We assume that the number of different
  $\spawn(q)$ operations appearing in $\D$ is finite.
\end{itemize}
\end{itemize}
Observe that the above definition ensures that the set of labels of
transitions is finite.

We are particularly interested in classes of systems when $Q$ is not finite.
This is the case when, for example, individual sub-processes execute
recursive procedures.
For that purpose, the transition system  $\Ss$ may be chosen as a
configuration graph of a \emph{pushdown
  system}. 
In this case the set $Q$ of states is $Q_l\cdot\Gamma^*$
where $Q_l$ is a finite set of  
control states, and $\Gamma$ is a finite set of pushdown symbols. 
The (infinite) transition relation $\Delta$ between states is specified
by a finite set of rewriting rules of the form 
$qv\act{a}q'w$ for suitable $q,q'\in Q_l,v\in\Gamma^*,w\in\Gamma^*$.
%

Instead of plain recursive programs, we could also allow
\emph{higher-order} recursive procedures, realized by higher-order
pushdown systems or even collapsible pushdown systems as considered,
e.g., in \cite{Ong06,HMOS08}.  Here, procedures may take other
procedures as arguments. 

\subsection{Multiset semantics}\label{s:multiset}

A dynamic parametric process is a transition system with labels of a
special form. 
As we have seen from the examples, such a transition system can be
provided either directly, or as the configuration graph of a machine, or as the
flow-graph of a program with procedure calls. 
In this subsection we provide the operational semantics of programs given
by such transition systems, where we interpret the operations on variables 
as expected, and the spawns as creation of sub-processes.
The latter operation will not create one
sub-process, but rather an arbitrary number of sub-processes.
There will be also a set of global variables to which every sub-process has
access by means of reads
and writes.

As a dynamic parametric process executes, sub-processes may change the
values of local and global variables and spawn new children.  The global state of
the entire process can be thus represented as a tree of sub-processes with the
initial process at the root.  Nodes at depth 1 are the sub-processes
spawned by the root; these children can also spawn sub-processes that
become nodes at depth $2$, etc, see e.g.,~Figure~\ref{fig:idea}(a).  
Every sub-process has a set of local variables, that can be read and
written by itself, as well as by its children.

A global state of a dynamic parametric process $\Ss$ has the form of a \emph{multiset
  configuration tree}, or \emph{m-tree} for short. An m-tree is
  defined recursively by 
\[
\begin{array}{lll}
t	&{::=}&	(q,\lambda,M)
\end{array}
\] where $q\in Q$ is a sub-process state, $\lambda:X\to V$ is a
valuation of (local) variables, and $M$
is a \emph{finite multiset} of m-trees.
We consider only m-trees of finite depth.
Another way to say this is to define
m-trees of depth at most $k$, for every $k\in\Nat$
\begin{eqnarray*}
  \mtrees_0	&=&	Q\times (X\to V)\times[]\\
  \mtrees_k	&=&	Q\times (X\to V)\times\mathcal{M}(\mtrees_{k-1})\qquad\mbox{for $k>0$}
\end{eqnarray*}
where for any $U$, $\mathcal{M}(U)$ is the set of all finite
multisubsets of $U$. 
Then the set of all m-trees is given by $\bigcup_{k\in\Nat} \mtrees_k$.

We use standard notation for multisets. A multiset $M$ over a universe $U$
is a mapping $M:U\to {\mathbb N}_0$. It is finite
if $\sum_{t\in U}M(t) <\infty$. A finite multiset $M$ may also be represented by 
$M= [n_1\cdot t_1,\ldots,n_k\cdot t_k]$ if $M(t_i) = n_i$ for $i=1,\ldots,k$ and $M(t)=0$ otherwise.
In particular, the empty multiset is denoted by $[]$.
For convenience we may omit multiplicities $n_i = 1$.
We say that $t\in M$ whenever $M(t)\geq 1$, and
$M\subseteq M'$ whenever $M(t)\leq M'(t)$ for all $t\in M'$. Finally,
$M+M'$ 
is the mapping with $(M+M')(t) = M(t)+M'(t)$ 
for all $t\in U$.
%
For convenience, we also allow the short-cut
$[n_1\cdot t_1,\ldots,n_k\cdot t_k]$ for 
$[n_1\cdot t_1]+\ldots+[n_k\cdot t_k]$, i.e., we allow also multiple occurrences of the same
tree in the list. Thus, e.g., $[3\cdot t_1,5\cdot t_2, 1\cdot t_1] = [4\cdot t_1,5\cdot t_2]$.

The \emph{semantics} of a dynamic parametric process $\Ss$ is a
transition system denoted $\sem{\Ss}$. The states of $\sem{S}$ are m-trees,
and the set of possible edge labels is: 
\begin{eqnarray*}
\S&=&  \{\tau\}\cup \set{\spawn}\times Q \, \cup\\
&&\set{\ai(x,v),\ao(x,v),\ar(y,v),\aw(y,v),\abr(y,v),\abw(y,v): \\
&&\qquad x \in X, \, y\in X\cup G,v\in V} \,.
\end{eqnarray*}
Notice that we have two new kinds of labels $\abr(y,v)$ and $\abw(y,v)$. These
represent the actions of child sub-processes on global variables $y \in G$,
or on the local variables $x \in X$
shared with the parent.

Throughout the paper we will use the notation\label{def:Sext}
\begin{equation*}
  \Sext	= \set{\ai(x,v),\ao(x,v),\ar(g,v),\aw(g,v): 
	x\in X, g\in G, v\in V}
\end{equation*}
\noindent
for the set of so-called \emph{external actions}. They are called
external because they concern either the global variables, or the
local variables of the parent of the sub-process. Words in $\Sext^*$
will describe the \emph{external behaviors} of a sub-process, i.e., the
interactions with the external world.

The initial state is given by $t_\init=(q_\init,\lambda_\init,[])$, 
where $\lambda_\init$ maps all locals to the initial value $v_\init$.
A transition between two states of $\sem{\Ss}$ (m-trees) $t_1\sact{a}_\Ss t_2$ is 
defined by induction on the depth of m-trees.
We will omit the subscript $\Ss$ for better readability.
The definition is given in Figure~\ref{f:semantics}.

\begin{figure*}[t]
\textbf{External transitions:}
  \[
  \begin{array}{lllll}
    (q_1,\lambda,M)&\sact{a}&(q_2,\lambda,M)	&
    \text{  if $q_1\act{a}q_2$}& \text{for $a\in\Sext$}\\
    (q,\lambda,M_1)&\sact{\abr(g,v)}&(q,\lambda,M_2)	&
    \text{  if } M_1\sact{\ar(g,v)} M_2&  \text{for $g\in G$}	\\
    (q,\lambda,M_1)&\sact{\abw(g,v)}&(q,\lambda,M_2)	&
    \text{  if } M_1\sact{\aw(g,v)} M_2&  \text{for $g\in G$}	\\
  \end{array}
  \]
\medskip

 \textbf{Internal transitions:}
 \[
   \begin{array}{llll}
     (q_1,\lambda,M)&\sact{\tau}&(q_2,\lambda,M)    		&\text{if $q_1\act{\tau}q_2$}\\
     (q_1,\lambda,M_1)&\sact{\spawn(p)}& (q_2,\lambda, M_2)    	
 			&\text{if $q_1\act{\spawn(p)}q_2$ and $M_2 = M_1 + [n\cdot(p,\l_\init,[])]$}\;
 				\text{for some}\;m\geq 0 \\
     (q_1,\lambda,M)&\sact{\aw(x,v)}&(q_2,\l',M)    
 			&\text{if $q_1\act{\aw(x,v)}q_2$ and $\lambda'=\lambda[v/x]$}\\
     (q_1,\lambda,M)&\sact{\ar(x,v)}&(q_2,\l,M)    
 			&\text{if $q_1\act{\ar(x,v)}q_2$ and $v=\lambda(x)$}\\
     (q,\lambda,M_1)&\sact{\abr(x,v)}&(q,\lambda,M_2)    
 			&\text{if $M_1\sact{\ai(x,v)}M_2$ and $v=\lambda(x)$}\\
     (q,\lambda,M_1)&\sact{\abw(x,v)}&(q,\lambda',M_2)    
 			&\text{if $M_1\sact{\ao(x,v)}M_2$ and $\lambda'=\lambda[v/x]$}\\
   \end{array}
 \]
 \medskip

Here, we say that
\[
  \begin{array}{lll}
    M_1& \sact{a}& M_2 \quad \text{for } a \in \Sext 
  \end{array}
\]
  if there is a multi-subset $M_1 = M'+ [n_1\cdot t_1,\ldots,n_r\cdot t_r]$ 
  (where the $t_i$ need not necessarily be distinct) and 
  executions $t_i\sact{\a_i a} t'_i$ for $i=1,\ldots,r$ for sequences
  $\a_i \in (\S \setminus \Sext) ^* $ 
  and 
  $M_2=M'+[n_1\cdot t'_1,\ldots,n_r\cdot t'_r]$.    
\caption{\label{f:semantics} Multiset semantics of dynamic
  parametric processes.}
\end{figure*}

\emph{External transitions} (cf.~Figure~\ref{f:semantics})
describe operations on external variables, be they
global or local. 
If the actions come from child sub-processes then for technical convenience
we add a bar to them.
Thanks to adding a bar, a label determines the rule that has been used
for the transition (This is important in Prop.~\ref{p:single-step}).
The values of global variables are 
not part of the program state. Accordingly, these operations therefore
can be considered as unconstrained input/output actions.

\emph{Internal} transitions may silently change the current state, spawn new sub-processes or update or read the topmost local variables 
of the process. The expression $\lambda[v/x]$ denotes the function $\lambda':X\to V$ defined by 
$\lambda'(x') = \lambda(x')$ for $x'\neq x$ and $\lambda'(x) = v$.
In the case of \textsf{spawn}, 
the initial state of the new sub-processes is given by the argument,
while the fresh local variables are initialized with the default value.
In the last two cases (cf.~Figure~\ref{f:semantics}) the external actions $\ai(x,v),\ao(x,v)$ of the child sub-processes
get relabeled as the corresponding internal actions
$\abr(x,v),\abw(x,v)$ on the local variables of the parent.

We write  $t_1\sact{\alpha} t_2$ for a sequence of transitions
complying with the sequence $\alpha$ of action labels. 
We have chosen the option to allow several child sub-processes to move in one step. 
While this makes the definition slightly more complicated, it simplifies
some arguments later.
Observe that the semantics makes the actions (labels) at the
top level explicit, while the actions of child sub-processes are  explicit only
if they refer to globals or affect the local variables of the parent.

\ignore{
\begin{remark}
  It is not obvious how to simulate multiple local variables 
  by a single variable. Observe that it would
  not be correct to simply take one variable whose value is a tuple of
  values. 
  Such a coding would not allow to implement a write operation where
  only one component is changed. 
  For this we would need something like test-and-set operation, but
  if we allowed this operation all the questions we consider
  in this paper would become undecidable.
  A solution is to keep the value of multiple local variables in a
  state of the sub-process and to use register to pass read/write commands
  as well as their acknowledgments from sub-processes. 
  Unfortunately, even though this protocol is intuitively easy, the
  proof of its correctness is not. 
  Since we do not need this result here, we do not present this
  reduction.
  The same remark applies to global variables.
\end{remark}
}

\subsection{Problem statement and main result}
In this section we define the reachability problem 
and state our main theorem: it says that the reachability problem is
decidable for dynamic parametric processes built upon an
\emph{admissible} class of systems\igw{changed}.  
The notion of admissible class will be 
introduced later in this section.
Before we do so, we introduce a
\emph{consistency} requirement for runs of parametric processes.  
In our semantics we have chosen not to constrain the operations on
global variables. Their values are not stored in the overall state. 
At some moment, though, we must require that sequences of read/write
actions on some global variable $y \in G$
 can indeed be realized via reading from and
writing to $y$.

\begin{definition}[Consistency]\label{def:consistency} Let $y\in G$ be
  a global variable. A sequence  $\a \in \Sext^*$ is
  \emph{$y$-consistent} if in the projection of $\a$ on operations on
  $y$, every read action $\ar(y,v)$ or $\abr(y,v)$ which is not the
  first operation on $y$ in $\a$ is immediately preceded either by
  $\ar(y,v),\abr(y,v)$ or by $\aw(y,v)$ or $\abw(y,v)$.  The first
  operation on $y$ in
  $\a$ can be either $\ar(y,v_\init),\abr(y,v_\init)$ or
  $\aw(y,v),\abw(y,v)$ for some $v$.
 
A sequence $\a$ is \emph{consistent} if it is $y$-consistent for every
variable $y\in G$. 
Let $\Consistent$ be the set of all consistent sequences.  
As we assume both $G$ and $V$ to be finite, this is a regular language.
\end{definition}

\noindent
Our goal is to decide reachability for dynamic
parametric processes.

\begin{definition}[Consistent run, reachability] A \emph{run}
of a dynamic parametric process $\Ss$ is a path in $\sem{S}$ starting in the initial
state, i.e., a sequence $\a$ such that $t_\init\sact{\alpha}_\Ss t$
holds.  If $\a$ is consistent, it is called a \emph{consistent run}.

The \emph{reachability problem} is to decide if for a given $\Ss$,  there is  a
consistent run of $\sem{\Ss}$ containing an external write or an
output action 
of some  distinguished value $\#$. 
\end{definition}

Our definition of reachability talks about a particular value of some variable,
and not about a particular state of the process. 
This choice is common, e.g., reaching a bad state may be
simulated by writing a 
particular value, that is only possible from bad states.
The definition admits not only external writes but also output actions because
we will also consider processes without external writes.

We cannot expect the reachability problem to be decidable 
without any restriction on  $\Ss$.
\igw{changed} Instead of considering a particular class of dynamic parametric
processes, like those  build upon pushdown systems, we will formulate mild
conditions on a class of 
such systems that turn out to be sufficient for deciding the
reachability problem.  
These conditions will be satisfied by the class of
pushdown systems, that is our primary motivation.
Still we prefer this more abstract approach for two reasons.
First, it simplifies notations.
Second, it makes our results applicable to other cases as, for example,
configuration graphs of higher-order pushdown systems with collapse.

In order to formulate our conditions, we require the notion of \emph{automata},
with possibly infinitely many states. 
An \emph{automaton} is a tuple:
\begin{equation*}
  \Aa=\struct{Q,\S,\D\incl Q\times\S\times Q, F\incl Q}
\end{equation*}
where $Q$ is a set of states, $\S$ is a finite alphabet, $\D$ is a transition
relation, and $F$ is a set of accepting states. Observe that we do not
single out an initial state.
Apart from the alphabet, all other components may be infinite sets.

We now define what it means for a class of automata to have
sufficiently good decidability and closure properties. 

\begin{definition}[Admissible class of automata]\label{def:admissible}
  We call a class $\Cc$ of automata \emph{admissible} if it has the
  following properties: 
  \begin{itemize}
    \item \emph{Constructively decidable emptiness check}:
      For every automaton $\Aa$ from $\Cc$ and every state $q$
      of $\Aa$, it is decidable if $\Aa$ has some path from $q$ to an
      accepting state, and if the answer is positive then the sequence of labels of
      one such path can be computed.
      \item \emph{Alphabet extension}: There is an effective
construction that given an automaton $\Aa$ from $\Cc$, and an
alphabet $\G$ disjoint from the alphabet of $\Aa$, produces the
automaton $\Aa\lloop\G$ that is obtained from $\Aa$ by adding a
self-loop on every state of $\Aa$ on every letter from $\G$. Moreover,
$\Aa\lloop\G$ also belongs to $\Cc$.

\item \emph{Synchronized product with
    finite-state systems}: There is an algorithm that from a given
  automaton $\Aa$ from $\Cc$ and a finite-state automaton
  $\Aa'$ over the same alphabet, constructs the \emph{synchronous
    product} $\Aa\times \Aa'$, that belongs to $\Cc$, too. The states of the product are pairs of
  states of $\Aa$ and $\Aa'$; there is a transition on some letter from
  such a pair if there is one from both states in the pair.
  A pair of states $(q,q')$ is accepting in the synchronous product iff 
  $q$ is an accepting state of $\Aa$ and $q'$ is an accepting state of $\Aa'$.
   \end{itemize}
\end{definition}

\noindent
There are many examples of admissible classes of automata.
The simplest is the class of finite automata\igw{changed}. 
Other examples are
(configuration graphs of)  pushdown automata, higher-order pushdown
automata with collapse, VASS with action labels, communicating
automata, etc. 


Given a dynamic parametric process $\Ss$, we obtain an automaton
$\Aa_\Ss$ by\igw{changed}\ 
declaring all states
final. That is, given the transition system
$\Ss=\struct{Q,G,X,V,\Delta,q_\init,v_\init}$ we set $\Aa_\Ss=\struct{Q,\S_{G,X,V},\Delta,Q}$, where $\S_{G,X,V}$ is
the alphabet of actions appearing in $\Delta$. The automaton $\Aa_\Ss$
is referred to as the \emph{associated automaton} of $\Ss$.
%
The  main result of this paper is:

\begin{theorem}\label{thm:main}
  Let $\Cc$ be an admissible class of automata. The reachability problem
  for  dynamic parametric processes with associated automata
  in $\Cc$,
  is decidable.
\end{theorem}
\noindent
\igw{changed}As a corollary, we obtain that the reachability problem is decidable
for \emph{pushdown dynamic parametric  processes}, that is where each
sub-process is a pushdown automaton. Indeed, in this case $\Cc$ is the
class of pushdown automata.
Similarly, we get decidability for dynamic parametric processes with
subprocesses being higher-order pushdown automata with collapse, and
the other classes listed above.


\section{Set semantics}\label{s:set}

The first step towards deciding reachability for dynamic parametric processes is to simplify the semantics. 
The idea of using a set semantics instead of a multiset semantics has
 already been suggested in
~\cite{DBLP:conf/lics/Kahlon08,DBLP:conf/cav/Durand-Gasselin15,DBLP:conf/concur/TorreMW15,DBLP:journals/jacm/EsparzaGM16}.
We adapt it  to our model, and show that the resulting set semantics
is equivalent to the multiset semantics --- at least as far as the reachability problem
is concerned. 
We conclude this section with several useful properties of runs of our
systems that are easy to deduce from the set semantics.

\emph{Set configuration trees} or \emph{s-trees} for short, are of the form
\begin{equation*}
  s::=(q,\lambda,S)
\end{equation*}
where $q\in Q$, $\lambda:X\to V$, and $S$ is a finite set of s-trees. 
As in the case of m-trees, we consider only \emph{finite} s-trees. In
particular, this means that s-trees necessarily have finite depth. Configuration trees of depth $0$
are those where $S$ is empty. 
The set $\strees_k$\label{def:strees} of s-trees of depth $k\geq 0$ 
is defined in a similar way as the set $\mtrees_k$
of multiset configuration trees of depth $k$.

With a given dynamic parametric process $\Ss$, the set semantics associates  a transition
system $\sem{\Ss}_s$ with s-trees as states. 
Its transitions have the same labels as in the case of multiset
semantics. Moreover, we will use the same notation as for multiset
transitions. It should be clear which semantics we are referring to,
as we use $t$ for m-trees and $s$ for s-trees.

As expected, the initial s-tree is $s_\init=(q_\init,\lambda_\init,\es)$.

The transitions are defined as in the multiset case but 
for multiset actions that become set actions:
  \begin{equation*}
    S\sact{\spawn(p)}S\cup\set{(p,\l_\init,\emptyset)}\qquad\text{and}\qquad
    S_1 \sact{a}S_2  \text{\quad if } a \in \Sext 
  \end{equation*}
for $S_2 = S_1\cup B$ where for each $s_2\in B$ there is some $s_1\in S_1$ so that
  $s_1\sact{\a a} s_2$ for some sequence $\a \in (\S\setminus
  \Sext)^*$. 
\medskip

The reachability problem for dynamic parametric processes under
the set semantics asks, like in the multiset case, whether there is
some consistent run of $\sem{\Ss}_s$ that contains an external write or an
output of a special value $\#$.

\begin{proposition}
   The reachability problems of dynamic parametric processes under
   the multiset and the set semantics, respectively, are equivalent.
\end{proposition}

We proceed to show that the set and the 
multiset semantics are equivalent in the context of
reachability.

On s-trees and sets of s-trees, we define inductively the preorder $\sqsubseteq$ by
\begin{itemize}
\item	$s\sqsubseteq s$;
\item	if $S\sqsubseteq S'$ then $(q,\lambda,S)\sqsubseteq (q,\lambda,S')$;
\item	if for all $s\in S$ there is some $s'\in S'$ with $s\sqsubseteq s'$, then $S\sqsubseteq S'$.
\end{itemize}
The relation $\sqsubseteq$ is reflexive and transitive, but not necessarily anti-symmetric. Thus, it defines
an equivalence relation on s-trees.
 
%


%
Every m-tree determines an s-tree by changing multisets to sets:
\begin{equation*}
  \sset((q,\lambda,M))=(q,\lambda,\set{\sset(t) : t\in M})
\end{equation*}

\noindent
The next two lemmas state a correspondence between multiset and set semantics.
\begin{lemma}\label{lemma:m-to-s}
  For all m-trees $t_1$, $t_2$, multisets $M_1$, $M_2$, s-tree $s_1$,
  and set of s-trees $S_1$:
  \begin{itemize}
  \item If $t_1\sact{a} t_2$ and $\sset(t_1)\leqtree s_1$ then $s_1\sact{a}
    s_2$ for some $s_2$ with $\sset(t_2)\leqtree s_2$.
  \item If $M_1\sact{a} M_2$ and $\sset(M_1)\leqtree S_1$ then $S_1\sact{a}
    S_2$ for some $S_2$ with $\sset(M_2)\leqtree S_2$.
  \end{itemize}
\end{lemma}
\begin{proof}
We will show only the most involved case of multiset
transitions. 
Suppose $M_1\sact{a} M_2$. 
Then we have by definition some subset $B= [n_1\cdot t_1,\ldots,n_r\cdot t_r]$ of $M_1$
where for $i=1,\ldots,r$, 
$t_i\sact{\alpha_i a} t'_{i}$ holds for a sequence $\alpha_i
\in(\S\setminus\Sext)^*$, and $M_2 = M_1\dminus B+[n_1\cdot t'_1,\ldots,n_r\cdot t'_r]$.
  Since $\sset(M_1)\leqtree S_1$, we have for all $i$,
some  $s_i\in S_1$ with $\sset(t_{i})\leqtree s_{i}$.
  Then by induction assumption $s_{i}\sact{\alpha_i a} s'_{i}$ with
  $\sset(t'_{i})\leqtree s'_{i}$.
  Taking  $S_2=S_1\cup\set{s'_{1},\ldots,s'_r}$ we obtain
  $\sset(M_2)\leqtree S_2$ and $S_1\sact{a} S_2$.
\qed
\end{proof}

\noindent
For the next lemma we introduce the auxiliary notions of 
\emph{$n$-thick} multisets and  \emph{$n$-thick} m-trees for $n\in
\Nat$. They are defined by mutual recursion. 
A multiset is $n$-thick if every element in $M$ is $n$-thick and
appears with the multiplicity at least $n$ (note that $M=[]$ is
$n$-thick for every $n$). 
An m-tree $t=(q,\lambda,M)$ is $n$-thick if $M$ is $n$-thick.

\begin{lemma}\label{lemma:s-to-m}
  \begin{itemize}
  \item If $s_1\sact{a} s_2$ for s-trees $s_1,s_2$, then there is 
    some factor $m\geq 1$ so that for every $n\geq 1$ and 
  $(m\cdot n)$-thick $t_1$ with $\sset(t_1)=s_1$, $t_1\sact{a} t_2$ holds for some $n$-thick  $t_2$ 
  with $\sset(t_2) = s_2$.
  \item If $S_1\sact{a} S_2$ for sets of s-trees $S_1,S_2$, 
  then there is some factor $m\geq 1$ so that for every $n\geq 1$ and
  $(m\cdot n)$-thick $M_1$ with $\sset(M_1)=S_1$, $M_1\sact{a} M_2$ holds for some $n$-thick
  multiset $M_2$ with $\sset(M_2)=S_2$.
  \end{itemize}
\end{lemma}
\begin{proof}
  We consider only set transitions. 
  If $a=\spawn(p)$, then $S_2=S_1\cup\{(p,\l_\init,\emptyset)\}$ and we can choose $m$ as $1$.
  Now assume that $S_2= S_1\cup \set{s'_1,\ldots,s'_{m'}}$ 
  where for each $i$, there is some $s_i\in S_1$ with
  $s_{i}\sact{\alpha_i a} s'_{i}$ for some sequence $\alpha_i \in (\S\setminus\Sext)^*$ of actions with corresponding factor $m_i$.
  Then define $m$ as the maximum of $m'+1$ and $m_i,i=1,\ldots,m'$.
  Consider some $M_1$ with $\sset(M_1)=S_1$ which is $(m\cdot n)$-thick. 
  Thus, for each $i$, there is some $t_i$ with $\sset(t_i)=s_i$ with $M_1(t_1)\geq mn\geq m_in$.
  The induction hypothesis gives us some $t'_i$ with $\sset(t'_i)=s'_i$ which is $n$-thick so that
  $t_{i}\sact{\alpha_i a} t'_{i}$. 
  We define $M_2= M_1\dminus[n\cdot t_1,\ldots,n\cdot t_{m'}]+[n\cdot t'_1,\ldots,n\cdot t'_{m'}]$.
  Then $M_2$ is $n$-thick and $M_1\sact{a} M_2$.
\qed
\end{proof}

\begin{corollary}\label{c:equivalence}
  The reachability problems for the multiset and set semantics are equivalent.
\end{corollary}
\begin{proof}
Lemma~\ref{lemma:m-to-s} implies that if $t_\init\sact{\a} t$ for some
sequence $\a$ and some $t$ then $s_\init\sact{\a} s$ for some $s$.

For the opposite direction take an execution $s_\init\sact{\a} s$ for
some $s$.
By definition, $t_\init$ is $m$-thick for every $m\geq 1$. 
Then Lemma~\ref{lemma:s-to-m} gives us an execution $t_\init\sact{\a} t$.
\qed
\end{proof}


\section{External sequences and signatures}

In this section we define some useful languages describing
the behavior of dynamic parametric processes.
Since our constructions and proofs will proceed by induction on the depth
of s-trees, we will be particularly interested in sequences of
external actions of subtrees of processes, \anca{changed}
and in signatures of such sequences, as defined below.
Recall the definition of the alphabet of external actions $\Sext$ 
(see page~\pageref{def:Sext}). 
Other actions of interest are the spawns occurring in $\Ss$: 
\[  \Ssp = \set{\spawn(p) : \text{$\spawn(p)$ is a label of a
  transition in $\Ss$}}\]
Recall that according to our definitions, $\Ssp$ is finite. 

For a sequence of actions $\a$, let $\ext(\a)$ be the subsequence of
external actions in $\a$, with additional renaming of $\abw$, and
$\abr$ actions to actions without a bar, if they refer to global
variables $g \in G$:
\[
\ext(a)\;=\;
\left\{
	\begin{array}{lll}
	\ar(g,v)	&\text{if}
			&a = \ar(g,v)\text{ or } a= \abr(g,v)	\\
	\aw(g,v)	&\text{if}
			&a = \aw(g,v)\text{ or } a= \abw(g,v)	\\
	a		&\text{if}
			&a = \ai(x,v)\text{ or } a= \ao(x,v)	\\
	\epsilon	&\multicolumn{2}{l}{\text{otherwise}}
	\end{array}\right .  
\]


%
Let $\sact{\alpha}_k$ stand for the restriction of $\sact{\a}$ to
s-trees of depth at most $k$ (the trees of depth $0$ have only the root).
This allows to define a family of languages of \emph{external
  behaviors} of trees of processes of height $k$. This family will be the main object of our study. 
\[
  \Ext_k=\set{\spawn(p)\,\ext(\alpha) : 
	(p,\l_\init,\emptyset)\sact{\alpha}_k s~\text{for some $s$, $\spawn(p)\in\Ssp$}}
\]

\ignore{Next we formulate two convenient properties of the set
semantics. These properties do not hold in the multiset
semantics (but some analogous version would hold if we worked with sufficiently
thick multisets). The first property says that in the set semantics we
can replicate (external) actions of sets:

\begin{lemma}\label{label:insert-out}
  Suppose $S\sact{\alpha a\beta} S'$ where $a$ is some action. Then
  for every decomposition $\beta=\beta_1\beta_2$ there is a
  computation $S\sact{\alpha a\beta_1a\beta_2} S'$. Similarly for
  $\sact{}_k$.
\end{lemma}
\begin{proof}
  Take
  \begin{equation*}
    S\sact{\alpha } S_1\sact{a}S_2\sact{\beta_1}  S_3\sact{\beta_2}  S'
  \end{equation*}
  We have $s_1\in S_1$, and $s_1\sact{\g a} s_2$  for some $\g
  \in(\S\setminus \Sext)^*$, 
  $S_2=S_1\cup\set{s_2}$. 
  Then $s_1\in S_3$, so $S_3\sact{a} S'_3$ with 
  $S'_3=S_3\cup\set{s_2}$. 
  But $S'_3=S_3$ since $s_2\in S_3$.
\end{proof}

\begin{lemma}\label{lemma:remove-spawn}
  Suppose $S \sact{\a\, \spawn(p)\, \b_1\, \spawn(p)\, \b_2} S'$ then
  $S\sact{\a\,\spawn(p)\,\b_1\b_2} S'$. Similarly for
  $\sact{}_k$.
\end{lemma}
\begin{proof}
  Consider $S \sact{\a} S_1 \sact{\spawn(p)} S_1\cup\set{(p,\l_\init,\es)}
  \sact{\b_1} S_2\sact{\spawn(p)} S_2\cup\set{(p,\l_\init,\es)}$.
  But by the definition of the set semantics, $(p,\l_\init,\es)\in S_2$, so the
  second $\spawn$ does not change the configuration, and can be eliminated.
\end{proof}
}

\noindent

The following definitions introduce abstraction and concretization operations
on (sets of) sequences of external actions.  The abstraction operation extracts
from a sequence its \emph{signature}, that is, the subsequence of 
first occurrences of external actions:

\begin{definition}[Signature, canonical decomposition]\label{def:signatures}
  The signature of a word $\a\in\Sext^*$, denoted $\sig(\a)$, is the
  subsequence of first 
  appearances of actions in $\a$.

  For a word $\a$ with signature $\sig(\a)=b_0b_1 \cdots b_k$, the
  \emph{(canonical)  decomposition} is $\a=b_0\a_1b_1\a_2 b_2\cdots  \a_kb_k\a_{k+1}$,
  where $b_i$ does not appear in $\a_1\cdots\a_i$, for all $i$.

For words $\b \in \Ssp \cdot \Sext^*$ the signature is defined by
$\sig(\spawn(p)\a)=\spawn(p) \cdot \sig(\a)$.
\end{definition}
The above definition implies that $\a_1$ consists solely of
repetitions of $b_0$. In Example~\ref{ex:prelim} the
signatures of the executions at level 1 are $\spawn(p) \ai(x,0)
\ao(x,1)$, $\spawn(p) \ai(x,1)
\ao(x,2)$, and $\spawn(p) \ai(x,1)
\ao(x,0)$. 
Observe that all signatures
at level 1 in this example are prefixes of the above
signatures.

While the signature operation removes actions from a sequence, the
concretization operation $\lift$ inserts them in all possible ways. 

\begin{definition}[lift]\label{df:lift}
  Let $\a\in\S_\ext^*$ be a word with signature $b_0,\dots,b_n$ and canonical decomposition
  $\a=b_0\a_1 b_1 \a_2 b_2\cdots \a_k b_k \a_{k+1}$. 
  A \emph{lift} of $\a$ is any word $\b=b_0\b_1 b_1 \b_2
  b_2\cdots \b_k b_k \b_{k+1}$, where $\b_i$ is obtained from
  $\a_i$ by inserting some number of actions $b_0,\dots,b_{i-1}$, for
  $i=1,\dots,k+1$. We write $\lift(\a)$ for the set of all such words
  $\b$. For a set $L \subseteq \Sext^*$ we define 
  \begin{equation*}
    \lift(L)=\bigcup\set{\lift(\a) : \a\in L}
  \end{equation*}
We also define $\lift(\spawn(p)\cdot \a)$ as the set $\spawn(p) \cdot
\lift(\a)$, and similarly $\lift(L)$ for $L\incl \Ssp\cdot \S_\ext^*$.
\end{definition}
Observe that $\a\in\lift(\sig(\a))$. Another useful property is that if $\b \in\lift(\a)$ then $\a,\b$ 
agree in their signatures.


\section{Systems under  hypothesis}\label{s:hypo}
This section presents the proof of our main result, namely,
Theorem~\ref{thm:main} stating that the reachability
problem for dynamic parametric processes is decidable for an
admissible class of systems.
The corresponding algorithm will analyze a process tree level by level. 
The main tool is an abstraction of child sub-processes by their external
behaviors. We call it \emph{systems under
  hypothesis}. 

\begin{figure}[tbh]
  \centering
  \includegraphics[scale=.5]{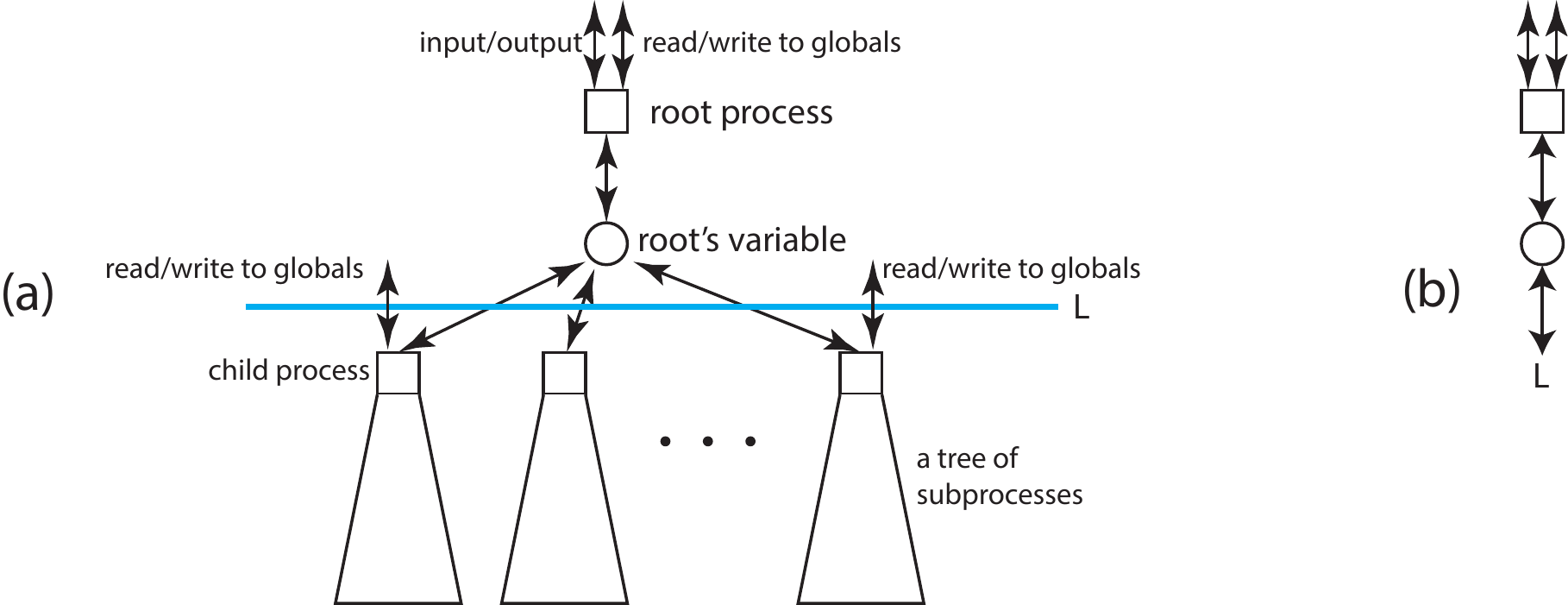}
  \caption{Reduction to a system under hypothesis.}
  \label{fig:idea}
\end{figure}
Let us briefly outline this idea.  A configuration of a
dynamic parametric process is a tree of sub-processes, Figure~\ref{fig:idea}(a).  The root 
performs (1) input/output external operations, (2) read/writes
to global variables, and (3) internal operations in form of
reads/writes to its local variables, that are also accessible to the
child sub-processes.  We are now interested in possible sequences of
operations on the global variables and the local variables of the root, that can be
done by the child sub-processes.  If somebody provided us with the set $L_p$ of all
such possible sequences, for child sub-processes starting at state $p$, for
all $p$, we could simplify our system as illustrated in
Figure~\ref{fig:idea}(b).  We would replace the set of all sub-trees
of the root by (a subset of) $L=\set{\spawn(p)\b:\b\in \pref(L_p), \spawn(p)\in\Ssp}$
summarizing the possible behaviors of child sub-processes.


A set  $L\incl \Ssp\cdot\Sext^*$ is called a
\emph{hypothesis}, as it represents a guess about the possible
behaviors of child sub-processes. 

Let us now formalize the notion of  \emph{execution of the system under
hypothesis}. 
For that, we define a system $\Ss_L$ that cannot spawn  child sub-processes, but
instead may use the hypothesis $L$.
We will show that if $L$ correctly describes the behavior of
child sub-processes then the set of runs of $\Ss_L$ equals the set of runs of $\Ss$
with child sub-processes. 
This approach provides a way to compute the set of possible external behaviors 
of the original process tree level-wise: 
first for systems restricted to s-trees of height at most $1$, then $2$, \dots,
until a fixpoint is reached.

The configurations of $\Ss_L$ are of the form $(q,\l,B)$, where $\l$ is  as
before a valuation of local variables, and 
$B\subseteq\pref(L)$ is a set of sequences of external actions for sets of sub-processes.

The initial state is  $r_\init=(q_\init,\l_\init,\emptyset)$.
We will use $r$ to range over configurations of $\Ss_L$. 
Transitions between two states $r_1\lact{a} r_2$ are listed in
Figure~\ref{F:hyp-L}.
Notice that transitions on actions of child sub-processes are modified so
that now $L$ is used to test if an action of a simulated child
sub-process is possible. 

\begin{figure*}[hbt]
\textbf{External transitions under hypothesis:}
\[
\begin{array}{lclll}
    (q_1,\l,B)&\lact{a}&(q_2,\l,B)&
    \text{if}& q_1\act{a}q_2 \;\text{if $a \in\Sext$}\\
    (q,\l,B)&\lact{\abw(g,v)}&(q,\l,B \cup B'\cdot\set{\aw(g,v)})    &\text{if}& \es\neq B'\incl B,\;
                                                          B'\cdot\{\aw(g,v)\}\incl
                                                          \pref(L)\\
     (q,\l,B)&\lact{\abr(g,v)}&(q,\l,B \cup B'\cdot\set{\ar(g,v)})
                                                &\text{if}& \es\neq
                                                B'\incl B, \;
                                                          B'\cdot\{\ar(g,v)\}\incl\pref(L)                                                         
  \end{array}
\]

\textbf{Internal transitions under hypothesis:}
\[
  \begin{array}{lclll}
    (q_1,\l,B)&\lact{\tau}&(q_2,\l,B)    &\text{if}& q_1\act{\tau}q_2\\
    (q_1,\l,B)&\lact{\spawn(p)}&(q_2,\l,B\cup\{\spawn(p)\})   &\text{if}&
                                                q_1\act{\spawn(p)}q_2\;
                                                      \text{and}\;
                                                      \spawn(p)\in\pref(L)\\
    (q_1,\l,B)&\lact{\aw(x,v)}&(q_2,\l',B)    &\text{if}&
						q_1\act{\aw(x,v)}q_2\;
						      \text{and}\;\l'=\l[v/x]\\
    (q_1,\l,B)&\lact{\ar(x,v)}&(q_2,\l,B)    &\text{if}&
						q_1\act{\ar(x,v)}q_2\; \text{and}\;\l(x)=v\\
    (q,\l,B)&\lact{\abw(x,v)}&(q,\l',B \cup B'\cdot\set{\ao(x,v)})
    &\text{if}& \es\neq B'\incl B, \;
                                                          B'\cdot\{\ao(x,v)\}\incl
                                                          \pref(L)\\
&&&& \l'=\l[v/x] \\
     (q,\l,B)&\lact{\abr(x,v)}&(q,\l,B \cup B'\cdot\set{\ai(x,v)})
                                                &\text{if}& \es\neq
                                                B'\incl B, \;
                                                          B'\cdot\{\ai(x,v)\}\incl
                                                          \pref(L),\\
&&&& \l(x)=v
  \end{array}
\]
\caption{\label{F:hyp-L}Transitions under hypothesis ($g\in G$, $x\in X$).}
\end{figure*}

\noindent
We list below two properties of $\lact{}$. In order to state them in a
convenient way, we introduce a \emph{filtering} operation
$\dist$ on sequences. 
The point is that external actions of child sub-processes are changed to
$\abr$ and $\abw$, when they are exposed at the root of a configuration tree.
In the definition below we rename them back; additionally, we remove
irrelevant actions. So $\dist(\a)$ is obtained by the following renaming of $\alpha$:
\[
\dist: \qquad 
\begin{array}{lll}
  \abr(x,v) \to \ai(x,v)\,, & \qquad & \abr(g,v) \to \ar(g,v)\,, \\
\abw(x,v) \to \ao(x,v) \,,& \qquad & \abw(g,v) \to \aw(g,v)\,,\\
a \to a & \qquad & \text{if }a \in\Ssp\,,\\
a \to \epsilon & & \text{otherwise}
\end{array}
\]


\noindent
The next two lemmas follow directly from the
definition of $\lact{\a}$.

\begin{lemma}\label{lemma:u-u-in-L}
 If $(q,\l,\es)\lact{\alpha} (q',\l',B)$ then $B\incl\pref(L)$,
 and every $\b\in B$ is a scattered subword of $\dist(\alpha)$. 
\end{lemma}

\begin{lemma}\label{l:mono}
  If $L_1\incl L_2$ and $(p,\l,\es)\liact{\a} r$ then
  $(p,\l,\es)\liiact{\a} r$.
\end{lemma}

\noindent
The next lemma states a basic property of the  relation $\lact{\a}$. If
we take for $L$ the set of all possible behaviors of child sub-processes with s-trees
of height at most $k$, then $\lact{\a}$ gives us all possible
behaviors of a system with s-trees of height at most $k+1$.
This corresponds exactly to the situation depicted in Figure~\ref{fig:idea}.

\begin{lemma}\label{lemma:u-two-systems}\label{l:hypo}
  Suppose $L=\Ext_k$. For every $p,q$, $\l$, and $\a$ we have:
  $(p,\l_\init,\es)\lact{\alpha} (q,\l,B)$ for some $B$ iff
  $(p,\l_\init,\es)\sact{\alpha}_{k+1} (q,\l,S)$ for some $S$.
\end{lemma}

\begin{proof}
  The statement may look almost tautological, but is not.
  We prove two directions:
  \begin{enumerate}
  \item		Whenever $(p,\l_\init,\es)\sact{\alpha}_{k+1} (q,\l,S)$,
		then there is some $B$ so that $(p,\l_\init,\es)\lact{\alpha} (q,\l,B)$ holds.
  \item		Whenever $(p,\l_\init,\es)\lact{\alpha} (q,\l,B)$, then there is some finite
		set $S$ of s-trees of depth at most $k$ so that 
		$(p,\l_\init,\es)\sact{\alpha}_{k+1} (q,\l,S)$.
  \end{enumerate}
  Statement 1 follows by induction on the length of $\alpha$ where the set $B$
  satisfies the invariant that for every $s\in S$,  there is some
  $\b=\spawn(p')\b'$  in $B$ such that $(p',\l_\init,\es)\sact{\a'}_{k} s$
  holds for some $\a'$, with $\ext(\a') =\b'$.

  The proof of statement 2 is more technical. Assume that $\alpha =
  a_1\cdots a_n$ and
  $(p_i,\l_i,B_i)\lact{a_i}(p_{i+1},\l_{i+1},B_{i+1})$ for
  $i=1,\ldots,n$, where $(p_1,\l_1,B_1)=(p,\l_\init,\es)$ and
  $(p_{n+1},\l_{n+1},B_{n+1}) = (q,\l,B)$.  We construct below a
  corresponding sequence of sets of $s$-trees $S_1,\ldots, S_{n+1}$
  with $(p_i,\l_i,S_i)\sact{a_i}_k(p_{i+1},\l_{i+1},S_{i+1})$,
  $i=1,\ldots,n$.

Since $B_{n+1}\incl\pref(\Ext_k)$, we have for every $\beta=\spawn(p_\b)\b'\in B_{n+1}$ 
  that $(p_\b,\l_\init,\es)\sact{\g_\b}_k(q_\b,\l,S_\b)$, for some $\g_\b$ with $\ext(\g_\b) = \b'$.
  This means for $\b' = b_{\b,1}\cdots b_{\b,n_\b}$ that $\g_\b$ has
  the form:
  $\g_\b = \g_{\b,1}b_{\b,1}\cdots\g_{\b,n_\b}b_{\b,n_\b}$ with $\ext(\g_{\b,j})=\e$ for all $j$.
  Moreover, there are s-trees $s_{\b,j}$ of depth at most $k$ so 
  that $s_{\b,j}\sact{\g_{\b,j}b_{\b,j}}_k s_{\b,j+1}$ for
  $j=1,\ldots,n_\b$ with $s_{\b,1}=(p_\b,\l_\init,\es)$ and
  $s_{\b,n_\b+1}=(q_\b,\l,S_\b)$.
  Then we define $S_i$ as the set of all $s_{\b,j}$ such that 
  $\spawn(p_\b)b_{\b,1}\cdots b_{\b,j-1}$ is a scattered subword of
  $\dist(a_1\cdots a_{i-1})$.

  It remains to prove that
  $(p_i,\l_i,S_i)\sact{a_i}_k(p_{i+1},\l_{i+1},S_{i+1})$ holds  for every $i=1,\ldots,n$.
  For that we perform a case distinction on action $a_i$. 
  If $a_i$ is of the form $\tau,\ar(g,v),\aw(g,v),\ai(x,v)$ or $\ao(x,v)$,
  then $B_{i+1} = B_i$, $S_{i+1}=S_i$ and the assertion holds.
  If $a_i$ is the action $\spawn(p')$, then $B_{i+1} = B_i\cup\set{\spawn(p')}$.
  Likewise, $S_{i+1} = S_i\cup\set{(p',\l_\init,\es)}$ in accordance with our claim.
  
  The most complicated case is when  $a_i$ is of one of the forms
  $\abr(y,v),\abw(y,v)$ for $y\in X\cup G$ and $v\in V$.  
  Then $B_{i+1}=B_i\cup B'\cdot\set{b}$ for $b=\dist(a_i)$ and some
  $B'\subseteq B_i$ with 
  $B'\cdot\set{b_i}\incl B_{n+1}\incl\pref(L)$. 
  Let $S$ be the set consisting of all  $s_{\b,j+1}$ such that
  $\spawn(p_\b)b_{\b,1}\cdots b_{\b,j}$ is a  scattered subword of
  $\dist(a_1\dots a_i)$ and $b_{\b,j}=b$.
  This set $S$ is not empty since $B'$ is not empty.
  We have $S_{i+1}=S_i\cup S$ by the definition of $S_i$ and $S_{i+1}$.
  Now for every $s_{\b,j+1}\in S$ we have by definition
  $s_{\b,j}\sact{\g_{\b,j}b_{\b,j}}_k s_{\b,j+1}$, and
  $b_{\b,j}=b$. Since $\spawn(p_\b)b_{\b,1}\cdots b_{\b,j-1}$ is a
  scattered subword of $\dist(a_1\cdots a_{i-1})$, we have 
  $s_{\b,j}\in S_i$. This shows that we  have indeed
  $(p_i,\l_i,S_i)\sact{a_i}_k(p_{i+1},\l_{i+1},S_{i}\cup S)$, with
  $S_i\cup S=S_{i+1}$.

\qed
\end{proof}

\noindent
The question we will pursue now is whether in the lemma above, we may replace
$\Ext_k$ with some simpler set and still get all computations of the
system of height $k+1$. 
Of central importance here is the following lemma saying that the lift
operation (cf.~Definition~\ref{df:lift}) does not add new behaviours.

\begin{lemma}\label{lemma:lift}\label{l:lift}
Assume that $L\incl\Ssp\cdot\Sext^*$ and $L' = \lift(L)$. Then
$(p,\l_\init,\es)\lact{\a}(q,\l,B)$ for some $B\incl \pref(L)$ iff
$(p,\l_\init,\es)\lpact{\a}(q,\l,B')$ for some $B'\incl \pref(L')$.
\end{lemma}

\begin{proof}
The left-to-right direction is obvious by monotonicity, since $L\incl L'$.

We focus on the right-to-left direction. The main idea is that the
first occurrences of actions in sequences from $B$ suffice to simulate
any sequence from $B'$.

Assume that $\alpha = a_1\cdots a_n$, and
$(p_i,\l_i,B'_i)\lpact{a_i}(p_{i+1},\l_{i+1},B'_{i+1})$ for $i=1,\ldots,n$
where $(p_1,\l_1,B'_1)=(p,\l_\init,\es)$ and $(p_{n+1},\l_{n+1},B'_{n+1})=(q,\l,B')$.
For $i=1,\ldots,n+1$, we define a subset $B_i\incl \pref(L)$ by 
$B_i = \set{\b\in \pref(L): \lift(\b)\cap B'_i\neq\es}$. 
By case distinction, we verify that indeed
$(p_i,\l_i,B_i)\lact{a_i}(p_{i+1},\l_{i+1},B_{i+1})$ holds for $i=1,\ldots,n$.
If $a_i$ is either $\tau,\ar(x,v),\aw(x,v)$ or in $\Sext$, then
$B'_{i+1} = B'_i$, $B_{i+1}=B_i$, and the assertion holds inductively.
If $a_i$ is the action $\spawn(p')$, then $B'_{i+1} =
B'_i\cup\set{\spawn(p')}$, $B_{i+1} = B_i\cup\set{\spawn(p')}$, and
the assertion holds, again by induction. 

It remains to consider the case where $a_i$ is 
$\abr(y,v)$ or $\abw(y,v)$, where $y\in X\cup G$, $v\in V$. Let
$b=\dist(a_i)$, and 
$B'_{i+1} = B'_i\cup B'\cdot\set{b}$ for some $\es\neq B'\incl B'_i$
with $B'\cdot\set{b}\incl\pref(L')$. We need to find some $B \subseteq
B_i$ such that $B_{i+1}=B_i \cup B \cdot \set{b}$.

Let $B=\set{\b \in B_i: \lift(\b) \cap B' \not=\es}$. We claim that
$B_i \cup B \cdot \set{b}=B_{i+1}$. To show $B_i
\cup B \cdot \set{b} \subseteq B_{i+1}$ we argue that for every $\b
\in B$, $\lift(\b b) \cap B'b \not=\es$. Conversely, let $\b'b \in
B'\set{b}$, and $\b \in B_{i+1}$ such that $\b'b \in \lift(\b)$.  We
need to show that $\b\in B_i \cup B \set{b}$. In particular, we know
that $\b \in
\pref(L)$. 
Either $\b$ is of the form $\b=\b_1 b$ and $\b_1$ has no
$b$. Thus, $\b' \in\lift(\b_1)$, so $\b_1 \in B$ and $\b \in B
\set{b}$. Or $\b=\b_1 b\b_2$ with $\b' \in \lift(\b)$. Since $\b'\in
B'_i$ we have in this case $\b \in B_i$ by definition.\qed

\end{proof}

\noindent
So Lemma \ref{l:hypo} says that child sub-processes can be abstracted
by their external behaviors. 
Lemmas \ref{l:mono} and \ref{l:lift} allow to abstract a set
$L$ of external behaviors 
by a subset $L_1 \subseteq L$,  as long as $L\incl\lift(L_1)$ holds. In the
following, we 
introduce a  
\emph{well-quasi-order} to characterize a smallest such subset, which we call \emph{core}.

\begin{definition}[Order, core]
  We define an order on $\Sext^*$ by $\a\fleq \b$ if
  $\b\in\lift(\a)$. This extends to an order on $\Ssp\cdot\Sext^*$:
  $\spawn(p) \a \fleq \spawn(q) \b$ if $p=q$ and $\a \fleq \b$.
  For a set $L\incl \Ssp\cdot\Sext^*$, we define $\core(L)$ as the set
  of \emph{minimal}   words in $L$ with respect to the relation $\fleq$.
\end{definition}

The following lemma states the most important property of the order $\fleq$:


\begin{lemma}\label{lemma:order}
  The relation $\fleq$ is a well-quasi-order on words with equal signature.
  Since the number of signatures is finite,
  the set $\core(L)$ is finite for every set $L\incl\Ssp\cdot\Sext^*$.
\end{lemma}

\begin{proof}
  We spell out what it means that $\a\fleq \b$, by expanding
  the definition of $\lift(\a)$. 
  First recall that if $\a\fleq\b$ then the two sequences have the
  same signatures.   
  Let $\a=\spawn(p)\a'$ and $\b=\spawn(p)\b'$,
  $\sig(\a')=\sig(\b')=b_1\cdots b_k$ for some $p$.
  Consider the canonical decompositions of $\a',\b'$:
   \begin{equation*}
     \a'=b_1\a'_1b_2\a'_2 \cdots  b_k\a_k\,,\qquad
     \b'=b_1\b'_1b_2\b'_2 \cdots  b_k\b_k\ .
   \end{equation*}
   We have $\a\fleq \b$ iff $\a'_i$ is a scattered subword of $\b'_i$,
   for every $i=1,\dots,k+1$. 
   Since being a scattered subword is a well-quasi-order relation, the
   lemma follows. \qed
\end{proof}

\noindent
Consider, e.g., the set $L=\Ext_1$ of all external behaviors of depth
1 
in Example~\ref{e:simple}. Then $\core(L)$ consists of the sequences:
\begin{equation*}
  \spawn(q)\,\aw(g_0,\#),\;\spawn(p)\,
  \ai(x,1)\ao(x,2)\ao(x,3), \;  \spawn(p)\, \ai(x,1) \ao(x,3)\ao(x,2)
\end{equation*}
together with all their prefixes (recall that $k$ in $\Ext_k$ refers to s-trees
of depth \emph{at most} $k$).

The development till now can be summarized by the following:
\begin{corollary}\label{cor:lift-core}
  For a set $L\incl \Ssp\cdot\Sext^*$, and $L'=\core(L)$:
  $(p,\l_\init,\es)\lact{\a}(q,\l,B)$ for some $B\incl L$ iff
  $(p,\l_\init,\es)\lpact{\a}(q,\l,B')$ for some $B'\incl L'$.
\end{corollary}
\begin{proof}
  Since $\core(L)\incl L$, the right-to-left implication follows
  by monotonicity. 
  For the other direction we observe that $L\incl \lift(\core(L))$, so
  we can use Lemma~\ref{lemma:lift} and monotonicity.
\qed
\end{proof}

Now we turn to the question of computing the  relation $\lact{\a}$ for a
\emph{finite} set $L$. For this we need our admissibility assumptions
from page~\pageref{def:admissible}. 

\begin{proposition}\label{p:single-step}
  Let $\Cc$ be an admissible class of automata, and let $\Ss$ be a
  transition system whose associated automaton is in $\Cc$.
  Suppose we have two sets $L,L'\incl\Ssp\cdot\Sext^*$  with $L\incl
  L'\incl\lift(L)$. Consider the set
  \begin{equation*}
    K=\set{\spawn(p)\ext(\alpha) :
      \spawn(p) \in\Ssp\text{ and }(p,\l_\init,\es)\lpact{\alpha}r', \text{for some $r'$}}
  \end{equation*}
  determined by $\Ss$ and $L'$.
  If $L$ is finite then we can compute the sets
  \begin{align*}
    \core(K)\quad\text{and}\quad\core(\set{\a \in K: \text{$\a$ consistent}})\ .
  \end{align*}
\end{proposition}

The proof of the above proposition works by augmenting the transition 
system $\Ss$ by a finite-state component taking care of the
valuation of local variables and of prefixes of $L$ that were used in
the hypothesis. The admissibility of $\Cc$ is then used to compute the
core of the language of the automaton thus obtained.

\begin{proof} 
By Lemmas~\ref{l:mono} and~\ref{l:lift}, the relations $\lact{\a}$ and $\lpact{\a}$
are the same.
Since $L$ is finite, there are only finitely many sets $B\incl L$.
Moreover, the number of valuations $\l$ is finite by our initial
definitions. 
This allows to construct a finite automaton $\Aa$, whose states are pairs
$(\l,B)$ and transitions are as those of $\lact{}$ but without the
first component:
\[
\begin{array}{lclll}
    (\l,B)&\act{\abw(g,v)}&(\l,B \cup B'\cdot\set{\aw(g,v)})    &\text{if}& \es\neq B'\incl B,\;
                                                          B'\cdot\{\aw(g,v)\}\incl
                                                          \pref(L),\;\\
     (\l,B)&\act{\abr(g,v)}&(\l,B \cup B'\cdot\set{\ar(g,v)})
                                                &\text{if}& \es\neq
                                                B'\incl B, \;
                                                          B'\cdot\{\ar(g,v)\}\incl\pref(L)
  \end{array}
\]
\[
  \begin{array}{lclll}
    (\l,B)&\act{\spawn(p)}&(\l,B\cup\{\spawn(p)\})   &\text{if $\spawn(p)\in\pref(L)$}\\
    (\l,B)&\act{\aw(x,v)}&(\l[v/x],B)    \\
    (\l,B)&\act{\ar(x,v)}&(\l,B)    &\text{if $\l(x)=v$}\\
    (\l,B)&\act{\abw(x,v)}&(\l[v/x],B \cup B'\cdot\set{\ao(x,v)})
    &\text{if } \es\neq B'\incl B, \;B'\cdot\{\ao(x,v)\}\incl\pref(L)\\
     (\l,B)&\act{\abr(x,v)}&(\l,B \cup B'\cdot\set{\ai(x,v)})
                                                &\text{if } \l(x)=v,
                                                \, \es\neq
                                                B'\incl B, \;
                                                          B'\cdot\{\ai(x,v)\}\incl
                                                            \pref(L)\\
  (\l,B)&\act{a}&(\l,B)&\text{for all } a \in \Sext
  \end{array}
\]

Now consider the automaton $\Aa_\Ss$ associated to $\Ss$. 
This automaton belongs to our admissible class $\Cc$, so its alphabet
extension is also in $\Cc$:
\begin{equation*}
\Aa'_\Ss=\Aa_\Ss\lloop\set{\abr(y,v),\abw(y,v) :   y\in X\cup G, v\in
  V}\ .
\end{equation*}
Intuitively, we add to $\Aa_\Ss$ self-loops on actions that are in $\Aa_L$ but not
in $\Aa_\Ss$. 
Finally, consider the product $\Aa^K=\Aa'_\Ss\times \Aa_L$.
We have that for every pair of states $q,q'$, valuations $\l, \l'$, and
sets $B,B'$: $(q,\l,B)\lact{\a}
(q',\l,B')$  iff there is a path labeled $\a$ from
$(q,\l,B)$ to $(q',\l,B')$ in $\Aa^K$.\igw{Here actions $\abw(g,v)$, $\abr(g,v)$
  are important}
Then $\spawn(p)\a\in K$ iff there is path in $\Aa^K$ from
$(p,\l_\init,\es)$  labeled by some $\a'$ with
$\ext(\a')=\a$. 

The above paragraph says that in order to compute $\core(K)$ it is
enough to compute $\core(\ext(K_p))$ where $K_p$ is the set of labels of
runs of $\Aa^K$ from $(p,\l_\init,\es)$.
Since $\Aa^K$ belongs to our admissible class $\Cc$, we can use
the effective emptiness test. 
If the language of $\Aa^K$ is not empty then 
we can find a word $\a_1$ that is accepted from $(p,\l_\init,\es)$.
Next we look for $\b\in \core(\ext(K_p))$ with $\b\fleq\ext(\a_1)$.
To this end, for every $\b\fleq \ext(\a_1)$ we consider an automaton
$\Aa^\b$ accepting all the words $\a'$ such that $\ext(\a')=\b$.
We build the product of $\Aa^K$ with $\Aa^\b$ and check for emptiness.
Then we choose one minimal $\b_1$ for which this product is non-empty.

To find a next word from $\core(ext(K_p))$ we construct a finite automaton
$\Nn_{\b_1}$ accepting all words $\a'$ such that $\b_1\not\fleq
\ext(\a')$. 
Then we consider $\Aa^K\times \Nn_{\b_1}$ instead of
$\Aa^K$. If the language accepted by $\Aa^K\times \Nn_{\b_1}$ is not
empty then we get a word $\a_2$ in the language. 
We apply the above procedure  to $\a_2$, iterating through all words $\b\fleq \ext(\a_2)$ and checking if 
the language of $\Aa^K\times \Aa^\b_{\b_1}$ is empty; here 
$\Aa^\b_{\b_1}$ is a finite automaton accepting all words $\a'$ such
that $\ext(\a')=\b$ and $\a'\in\Nn_{\b_1}$. 
We choose one minimal $\b_2$ for which such a product is non-empty.
For the following iteration we  construct
$\Nn_{\b_2}$ accepting all words $\a'$ such that $\b_2\not\fleq
\ext(\a')$.
So $\Nn_{\b_1}\times\Nn_{\b_2}$ accepts all words $\a'$ that 
$\b_1\not\fleq \ext(\a')$ and $\b_2\not\fleq \ext(\a')$.
We continue this way, finding words $\b_1,\dots,\b_k\in
\core(\ext(K_p))$ till $\Aa^K\times
\Nn_{\b_1}\times\dots\times\Nn_{\b_k}$ is empty.  
At that point we know that $\set{\b_1,\dots,\b_k}=\core(\ext(K_p))$.

This procedure works also for the second statement by observing that
the set of all consistent sequences, let us call it $\Consistent$, is a
regular language. So instead of starting with $K$ in the above
argument we start with $K\cap\Consistent$.
\qed


\end{proof}

In the next two corollaries, $\Ss$ is such that its associated
automaton $\Aa_\Ss$ belongs to an admissible class.

\begin{corollary}\label{cor:base}
  The sets $\core(\Ext_0)$ and $\core(\Ext_0\cap\Consistent)$ are
  computable.
\end{corollary}

\begin{proof}
As we are concerned with  s-trees of depth 0, all occurring 
configurations are of the form $(q,\l,\es)$. 
This means that $\spawn(p)\a\in \Ext_0$ iff
$(p,\l_\init,\es)\lzact{\a}_{\es} (q,\l,\es)$ for some $q$ and $\l$.
We can then use Proposition~\ref{p:single-step} with $L=L'=\es$.
\qed
\end{proof}

\begin{corollary}\label{c:core}
  Under the hypothesis of Proposition~\ref{p:single-step}: for every $k\geq 0$, we can compute $\core(\Ext_k)$ and
  $\core(\Ext_k\cap\Consistent)$. 
\end{corollary}

\begin{proof}
  We start with $L_0=\core(\Ext_0)$ that we can compute by
  Corollary~\ref{cor:base}. 
  Now assume that $L_i=\core(\Ext_i)$ has already been computed.
  By Lemma~\ref{l:hypo}, $L_{i+1}$ equals
  the core of
  $\set{\spawn(p)\ext(\a):(p,\l_\init,\es)\lIact{\a}r,\;\text{some $r$}}$
  which, by Proposition~\ref{p:single-step}, is effectively computable.
\qed
\end{proof}

\noindent
Now we have all ingredients to prove Theorem~\ref{thm:main}.

\begin{proof}[of Theorem~\ref{thm:main}]
  Take a process $\Ss$ as in the statement of the theorem.
  The external behaviors of $\Ss$ are described by the language
  \begin{equation*}
    L=\bigcup_{k\in\Nat}\set{\spawn(p)\ext(\a) : (p,\l_\init,\es)\lzact{\a}_{\Ext_k} r,\text{ for some } r}   
  \end{equation*}
  If we denote $L_k=\core(\Ext_k)$ then by
  Corollary~\ref{cor:lift-core}, the language $L$ is equal to
  \begin{equation*}
    L'=\bigcup_{k\in\Nat}\set{\spawn(p)\ext(\a) : (p,\l_\init,\es)\lzact{\a}_{L_k} r,\text{ for some } r}   
  \end{equation*}
  By definition, $\Ext_0\incl \Ext_1\incl\cdots$ is an increasing
  sequence of sets.
  By Lemma~\ref{lemma:order}, this means that  there is some $m$ so that
  $\core(\Ext_m)=\core(\Ext_{m+i})$, for all $i$. Therefore, $L'$ is
  equal to
  \begin{equation*}
    \set{\spawn(p)\ext(\a) : (p,\l_\init,\es)\lzact{\a}_{L_m} r,\text{ for some } r}   
  \end{equation*}
  By Corollary~\ref{c:core}, the set $L_m=\core(\Ext_m)$ is computable and so is
  $\core(L'\cap \Consistent)$. 
  Finally, we check if in this latter set there is a sequence starting with
  $\spawn(q_\init)$  and an external write or an output of
  $\#$.
\qed
\end{proof}


\section{Processes with generalized features}

In this section we consider \emph{pushdown} dynamic parametric  processes
where a sub-process cannot write to its
own variables, but only to the variables of its parent.
This corresponds to the situation when after a parent has created sub-processes,
the latter may communicate computed results to the parent
and to their siblings, but the parent cannot
communicate to child sub-processes. 
We call such a model a \emph{pushdown dynamic parametric process with
  generalized futures}.
We have seen an example of such a system in Figure~\ref{f:example}.
Technically, processes with generalized futures are obtained by disallowing $\aw(x,v)$ actions
in our general definition.
Additionally, we rule out global variables, i.e., $G=\es$.
Accordingly, we may no longer define reachability via reachability of a write action to some 
global variable, but as reachability of an output action $\ao(x,\#)$
of some special value $\#$ to some variable $x$ of the root process.

For processes with generalized futures, reachability can be decided by a somewhat simpler approach. 
In particular, we present an \EXPTIME\ algorithm to decide reachability.

In this section we need an additional assumption concerning the
initial value of variables.  Since this initial value causes problems as it
is the one that cannot be reproduced once overwritten, we require:\\

\noindent\textbf{Proviso:}
We consider systems where the initial value $v_\init$ of a variable can
  be neither read nor written. \igw{Important new thing}

\medskip

Since in the case that we consider in this section
there are no global variables, the external alphabet simplifies to:
\[
\Sext= \set{\ai(x,v),\ao(x,v) : x\in X, v\in V}\,.
\]
From the definition of transitions $r_1\lact{a} r_2$ of sub-processes
modulo hypothesis we can see that the label $a$ can be either an
external action, or internal action of the form:
$\t,\spawn(p),\ar(x,v),\abw(x,v),\abr(x,v)$.

Disallowing $\aw(x,v)$ operations has an important impact on the
$\lact{\a}$ semantics.  
By inspecting the rules, we notice that there is only one remaining rule, 
namely $\abw(x,v)$, that changes the value of the component $\l$, and this
rule does not change the state component. 
%


\noindent
The next lemma is the main technical step in this section.  It says
that for processes with generalized futures, $\sig(\Ext_k)$ as
hypothesis yields the same behaviors as $\Ext_k$.

\begin{lemma}\label{lem:sig-equal}
  For a dynamic parametric process $\Ss$ with generalized futures, let
  $L=\sig(\Ext_k)$ and $L' = \Ext_k$. We have:
$(p,\l_\init,\es)\lact{\a}(q,\l,B)$ for some $\l$, $B\incl L$ iff
$(p,\l_\init,\es)\lpact{\a'}(q,\l',B')$ for some $\l'$, $B'\incl L'$, and some
$\a'$ with $\ext(\a)=\ext(\a')$. 
\end{lemma}

\begin{proof}
 For the right-to-left implication observe that
  $\Ext_k\incl\lift(\sig(\Ext_k))$.
  So, if $(p,\l_\init,\es)\lpact{\a}(q,\l,B)$ then 
  $(p,\l_\init,\es)\lppact{\a}(q,\l,B)$ for 
  $L''=\lift(\sig(\Ext_k))$.
  But then Lemma~\ref{lemma:lift} gives us
  $(p,\l_\init,\es)\lact{\a}(q,\l,B')$ for some $B'$.

  For the left-to-right implication note first that it does not follow from
  monotonicity, since $\sig(\Ext_k)$ may contain sequences that are
  not in $\Ext_k$.  So assume that $\alpha = a_1\cdots a_n$, and
  $(p_i,\l_i,B_i)\lact{a_i}(p_{i+1},\l_{i+1},B_{i+1})$ for
  $i=1,\ldots,n$ where $(p_1,\l_1,B_1)=(p,\l_\init,\es)$ and
  $(p_{n+1},\l_{n+1},B_{n+1})=(q,\l,B)$.  We look for subsets $B'_i$
  of $\Ext_k$, for $i=1,\dots,n+1$, such that
  $(p_i,\l'_i,B'_i)\lpact{\d_ia_i} (p_{i+1},\l'_{i+1},B'_{i+1})$ for
  some sequence $\d_i$ of internal actions.  The additional $\d$'s are
  needed since the elements of $L$ are subsequences of those  from $L'$.

  For every \emph{maximal} signature $\b\in B_{n+1}\incl\pref(L)$, we
  fix a sequence $\b'\in \pref(L')$
  such that $\sig(\b')=\b$: for
  $\b=\spawn(p_\b)b_{\b,1}\cdots b_{\b,n_\b}$ we will write
  $\b'=\spawn(p_\b)b_{\b,1}\g_{\b,1}\cdots
  b_{\b,n_\b}\g_{\b,n_\b}$, for some $\g_{\b,j}\in \Sext^*$.
  
  We define $B'_{i}$ as the smallest prefix closed set that satisfies
  the following property for
  every maximal sequence $\b\in  B_{n+1}$:  if the prefix $\spawn(\b)b_{\b,1}\cdots
  b_{\b,j}$ is in $B_i$ (for
  some $j\leq n_\b$) then    $\spawn(p_\b)b_{\b,1}\g_{\b,2}\cdots
  b_{\b,j-1}\g_{\b,j-1}b_{\b,j}$ 
  is in $B'_i$.

  It remains to prove that
  $(p_i,\l'_i,B'_i)\lpact{\d_ia_i} (p_{i+1},\l'_{i+1},B'_{i+1})$ for some
  sequence $\d_i$ of internal actions.
  For that we perform a case distinction on the action $a_i$.
    If $a_i$ is $\tau$, $\ai(x,v)$ or $\ao(x,v)$, then only the state changes,
  so the assertion holds by induction.
  If $a_i$ is $\spawn(p')$, then $B_{i+1} = B_i\cup\set{\spawn(p')}$.
  Likewise, $B'_{i+1} = B'_i\cup\set{\spawn(p')}$ according to our
  claim.

The case $a_i=\ar(x,v)$ is a bit more tricky. If $\l'_i(x)=v$ then we
can do $(p_i,\l'_i,B'_i)\lpact{a_i} (p_{i+1},\l'_{i+1},B'_{i+1})$ immediately. If
not, we need to re-establish the value $v$ for $x$. Since by our
assumption the initial value is neither read nor written, we know that there must be
some $j<i$ with $a_j=\abw(x,v)$. So $B'_i$ contains some sequence
ending with $\ao(x,v)$. But then we can ``replay'' this output and
do $a_i$: 
$(p_i,\l'_i,B'_i)\lpact{\abw(x,v)} (p_i,\l''_i,B'_i) \lpact{\ar(x,v)}
(p_{i+1},\l''_i,B'_i)$. 
  
  The last, more involved case is when  $a_i \in
  \set{\abr(x,v),\abw(x,v): x\in X, v\in V}$.
  Then $B_{i+1}=B_i\cup B\cdot\set{b_i}$ for $b_i=\dist(a_i)$ and some
  non-empty $B\subseteq B_i$ with 
  $B\cdot\set{b_i}\incl B_{n+1}$. 
  Consider the set
  \begin{equation*}
    B'=\set{\spawn(p_\b)b_{\b,1}\g_{\b,2}\cdots b_{\b,j} :
      \spawn(p_\b)b_{\b,1}\cdots b_{\b,j}\in B}\,.
  \end{equation*}
  By definition, $B'\incl B'_i$, and $B'$ is not empty since $B$ is
  not empty.  Take some element of $B'$, say
  $\b'=\spawn(p_\b)b_{\b,1}\g_{\b,2}\cdots b_{\b,j}$.  For
  $\g=\g_{\b,j+1}$ we have $\b' 
  \g\in \pref(L')$ and $\sig(\b'\, \g)=\sig(\b')$.  We claim that we can construct a run
  $(p_i,\l'_i,\bar{B}_1)\lpact{\d}(p_i,\l''_i,\bar{B}_2)$, for every
  $\bar{B}_1$ containing $\b'$, and $\bar{B}_2$ consisting of
  $\bar{B}_1$ and all prefixes of $\b'\g$. The sequence $\d$ consists
  of internal actions.

  Assuming this claim, that we prove in the next paragraph,
  we proceed as above for each  sequence in $B'$, one after the other. 
  Since $B'$ is finite, say with $k$ elements, at the end we get a computation 
  $(p_i,\l'_i,B'_i)\lpact{\d_1\cdots\d_k}(p_i,\l''_i,B''_i)$ with
  $B''_i$ consisting of $B'_i$ and all prefixes of the set
  \begin{equation*}
    B''=\set{\spawn(p_\b)b_{\b,1}\g_{\b,2}\cdots b_{\b,j-1}\g_{\b,j} :
      \spawn(p_\b)b_{\b,1}\cdots b_{\b,j-1}\in B}\,.
  \end{equation*}
  Observe that $B''_i\cup B''\cdot\set{b_i}=B'_{i+1}$.
  Then we can do
  $(p_i,\l'_i,B''_i)\lpact{a_i}(p_{i+1},\l'_{i+1},B'_{i+1})$ as
  claimed at the beginning.

  It remains to prove the claim from the above paragraph. 
  Consider a configuration $(q,\l',\bar{B})$ with $\bar{B} \incl \pref(L')$ a prefix
  closed set, 
  and a sequence $\b'\in \bar{B}$. We want to show that for every sequence
  $\g$ of external
  actions such that $\b'\g\in \pref(L')$ and 
  $\sig(\b')=\sig(\b'\g)$, 
  there is a sequence $\d$ of internal actions such that 
  $(q,\l',\bar{B})\lact{\d}(q,\l'',\bar{B}_1)$, with $\bar{B}_1$ consisting of $\bar{B}$ and all
  prefixes of $\b'\g$. 
  
  Let $\g=c_1\cdots c_k$.
  Since $\sig(\b')=\sig(\b'\g)$, for every $c_i$ there is a prefix of
  $\b'$ ending in $c_i$, say  $\b'_ic_i$.
  If $c_i$ is an input action, say $\ai(x,v)$, then we must have that
  at the moment when $\b'_ic_i$ was added into the $\bar{B}$ component, the
  valuation $\l$ was such that $\l(x)=v$. 
  Thanks to our assumption that the initial value cannot be read,
  $v$ is not an initial value.
  So the only way to have $\l(x)=v$ was to execute an
 output $\ao(x,v)$ before. 
  This action gives a sequence $\b''_i\, \ao(x,v)\in \bar{B}$, for some
  $\b''_i$; in particular $\b''_i\,\ao(x,v)\in \pref(L')$.

  After these preparations we show how to execute the sequence $\g$.
  If $c_i$ is an output action then we just execute it since this is always
  possible as $\b'_ic_i\in \bar{B}$, and so $\b'_i\in \bar{B}$, by prefix closure.
  If $c_i$ is an input action, say $\ai(x,v)$, then we first execute
  $\ao(x,v)$, that si possible since $\b'_i\ao(x,v)\in \bar{B}$.
  Then we execute $\ai(x,v)$.
\qed
\end{proof}

Next we turn the question how to compute the set of signatures of
executions efficiently.

\begin{lemma}\label{lem:one-step}
  Let $\Ss$ be a pushdown dynamic parametric process with generalized futures.  For any
  state $p$ of $\Ss$, and any prefix closed set $L \subseteq \Ssp \Sext^*$ of
  signatures, $L=\sig(L)$, consider the language
  \begin{equation*}
    K_p=\set{\ext(\alpha) :
    (p,\l_\init,\es)\lact{\a}(q,\l,B) \text{ for some } B}\ .     
  \end{equation*}
  We  can compute in \EXPTIME\ the sets
  \begin{equation*}
    \set{\sig(\a) : \a\in K_p}\; ,\qquad \set{\sig(\a) : \a\in K_p\cap\Consistent}\,.
  \end{equation*}
    The computation time is exponential in $|V|$, $|X|$,
   and polynomial in the size of $p$ and the pushdown automaton defining $\Ss$.
\end{lemma}

\begin{proof} 
  Assume that $\alpha = a_1\cdots a_n$, and
  $(p_i,\l_i,B_i)\lact{a_i}(p_{i+1},\l_{i+1},B_{i+1})$ for
  $i=1,\ldots,n$ where $(p_1,\l_1,B_1)=(p,\l_\init,\es)$ and
  $(p_{n+1},\l_{n+1},B_{n+1})=(q,\l,B)$.
  Recall that the configurations of $\lact{}$ are $(q,\l,B)$ where $q$ is the
state of the system, $\l:X\to V$ is a valuation of the internal
variables, and $B \subseteq \pref(L)$ is a set of words in $\Ssp
\Sext^*$. 
Since $L$ is a set of signatures its size is at most exponential in
$|V|$ and $|X|$.
So there are at most doubly exponentially many $B$'s. 
This is too much, as we are after a single
exponential complexity.
We will show that the same sequence can be executed with $B$'s of
polynomial size. 

Let $b_1\cdots b_k$ be the subsequence of first occurrences of actions
of the form $\abr(x,v)$, $\abw(x,v)$ in $a_1\cdots a_n$.
For every $b_j$ we choose a sequence $\b_j \in B$ that is used to perform
$b_j$.
In other words, if the first occurrence of $b_j$ is $a_l$ then $\b_j\in
B_l$ and $\b_j b_j\in B_{l+1}$. 
We define $B'_i$ as the subset of $B_i$ consisting
all the prefixes of words $\b_j$ that are in $B_i$; more precisely for
every $i=1,\dots,n$ we set:
\begin{equation*}
  B'_i=\set{\b \in B : \b  \text{ prefix of some } \b_j, \, j=1,\dots,k}
\end{equation*}
By induction on $i$ we show that
$(p_i,\l_i,B'_i)\lact{a_i}(p_{i+1},\l_{i+1},B'_{i+1})$, for $i=1,\dots,n$.
So every sequence of actions can be executed using configurations
where the $B$-component is of polynomial size. 

For every
$\spawn(p)\in\Ssp$, from the pushdown automaton defining $\Ss$ we 
can construct a pushdown automaton $\Pp^p$ for the language
\begin{equation*}
K^p=\set{\alpha :   (p,\l_\init,\es)\lact{\a}(q,\l,B)}) .  
\end{equation*}
For this we take the product of the automaton for $\Ss$ with a
finite-state automaton taking care of the $\l$-component and the
$B$-component. 
Note that thanks to the previous paragraph, this finite state
automaton is of exponential size
(the $\l$-component is exponential in $N=\max(|X|,|V|)$,
and the sets $B$ consist of at most $N^2$ sequences and their prefixes). 

The statement of the lemma requires us to compute the set
$\sig(\ext(K^p))$. This can be done by the following general algorithm 
starting with $i=0$ and $K_0=K^p$:
\begin{enumerate}
\item Find a word $\alpha_i\in K_i$. Let $\b_i=\sig(\ext(\a_i))$.
\item Consider $K_{i+1}=K_i\cap N_i$ where $N_i$ is the set of
  words $\a$ such that $\sig(\ext(\a))$ is not $\b_i$ ($N_i$ is a
  regular language). 
\item if $K_{i+1}$ not empty, go to the first step.
\end{enumerate} 
This iteration terminates in exponential number of iterations, since
there  are exponentially many signatures, and after every iteration
we find one new signature. 
Each iteration of the above algorithm takes exponential time: this is
because the pushdown automaton for $\Ss$ is of  exponential-size, and so is
every pushdown automaton obtained by taking consecutive intersections with the
finite automata for $N_i$.
Accordingly, the set $\sig(\exp(K^p))$ can be 
computed in \EXPTIME\ for $\Ss$ given by a pushdown system.
 
The statement concerning the consistent sequences follows by the same
argument using the fact that the set $\Consistent$ of consistent
sequences is regular. So it is enough to start the iteration from
$K^p\cap\Consistent$ instead of $K^p$.
\qed
\end{proof}

\begin{lemma}\label{lem:sigk} 
  Let $\Ss$ be a pushdown dynamic parametric process with generalized futures.
  Then for every $k\geq 0$,
  the sets $\sig(\Ext_k)$ as well as $\sig(\Ext_k\cap\Consistent)$ are
  computable in \EXPTIME\ (exponential in $|V|$, $|X|$, 
polynomial in the pushdown automaton defining
  $\Ss$, and independent of $k$).
\end{lemma}

\begin{proof}
  We start with $L_0=\sig(\Ext_0)$. We can compute it in \EXPTIME\
  thanks to Lemma~\ref{lem:one-step} by taking $L=\es$.

Assume now that $L_k=\sig(\Ext_k)$ is already computed, and set
$L'_k=\Ext_k$. By
Lemmas~\ref{lemma:u-two-systems} and~\ref{lem:sig-equal},
$\Ext_{k+1}$ is equal to  
  \begin{align*}
 & \set{\spawn(p) \ext(\a) :  (p,\l_\init,\es)\lzact{\a}_{L'_k}
  (q,\l,B) \text{ for some } B} \; =\\
& \set{\spawn(p) \ext(\a) :  (p,\l_\init,\es)\lzact{\a}_{L_k}
  (q,\l,B) \text{ for some } B}
  \end{align*}
We can now apply Lemma~\ref{lem:one-step} with $L=L_k$ to compute
$L_{k+1}=\sig(\Ext_{k+1})$. 

  Every step can be done in deterministic exponential time.
  The complexity bound follows since the number of steps is at most
  exponential:  the sets increase after each iteration and the number
  of signatures is exponential in $|X|,|V|$ (and polynomial in the
  number of states of the automaton defining $\Ss$).
\qed
\end{proof}

\begin{theorem}\label{th:sig}
  The reachability problem for pushdown dynamic parametric processes with generalized futures  is in \DEXPTIME\ (exponential in
  $|V|$, $|X|$, and polynomial in the size of the pushdown
  automaton defining $\Ss$.)
\end{theorem}

\begin{proof}
  Let $\Ss$ be a pushdown dynamic parametric process with generalized futures. 
  The external behaviors of $\Ss$ are described by the language
  \begin{equation*}
    L=\bigcup_{k\in\Nat}\set{\spawn(p)\ext(\a) :
      (p,\l_\init,\es)\lzact{\a}_{\Ext_k} r,\text{ for some } r}\ .
  \end{equation*}
  If we denote $L_k=\sig(\Ext_k)$ then by
  Lemma~\ref{lem:sig-equal} the language $L$ is equal to
  \begin{equation*}
    L'=\bigcup_{k\in\Nat}\set{\spawn(p)\ext(\a) :
      (p,\l_\init,\es)\lzact{\a}_{L_k} r,\text{ for some } r}\ .   
  \end{equation*}
  By definition, $\Ext_0\incl \Ext_1\incl\cdots$ is an increasing
  sequence of sets.
  As there are finitely many signatures, this means that  there is
  some $m$ so that 
  $\sig(\Ext_m)=\sig(\Ext_{m+i})$, for all $i$. Therefore, $L'$ is
  equal to
  \begin{equation*}
    L''=\set{\spawn(p)\ext(\a) : (p,\l_\init,\es)\lzact{\a}_{L_m} r,\text{ for some } r}   
  \end{equation*}
  %
  By Lemma~\ref{lem:sigk}, the set $L_m=\sig(\Ext_m)$ can be computed in
  \EXPTIME, and so can be $\sig(L''\cap \Consistent)$.
  Finally we check if in the latter language there is a sequence starting with
  $\spawn(q_\init)$  and containing $\ao(x,\#)$.
\qed
  \end{proof}

\subsection{Simple futures}

Here we consider pushdown processes with \emph{simple futures},
where every sub-process communicates only with its parent, by
writing values into the registers shared with the parent. So there is
no communication between siblings, therefore  we work with
$\Sext= \set{\ao(x,v) : x\in X, v\in V}$. Internal transitions are
with actions of the form $\t,\spawn(p),\ar(x,v),\abw(x,v)$.

For $\b \in \Sext^*$ let $\out(\b)=\set{\ao(x,v) : 
 \b=\spawn(p') \b_1 \ao(x,v) \b_2}$, and let $\out(\spawn(p)
\a)=\set{\spawn(p)} \cup \spawn(p) \out(\a)$. For a set $L
\subseteq \Ssp \cdot \Sext^*$ let $\out(L)=\set{\out(\b)
  : \b \in L}$.

\begin{lemma}
  Let $L \subseteq \Sext^*$ be a prefix-closed set of signatures and
  $L'=\out(L)$. Then $(p,\l_\init,\es)\lact{\a}(q,\l,B)$ for some
  $\l$ and $B\incl L$ iff $(p,\l_\init,\es)\lpact{\a'}(q,\l',B')$ for some
  $\l'$ and $B'\incl L'$, with $\ext(\a)=\ext(\a')$.
\end{lemma}

\begin{proof}
For the left-to-right direction we show by induction that
$(p,\l_\init,\es)\lact{\a}(q,\l,B)$ implies
$(p,\l_\init,\es)\lpact{\a}(q,\l,B')$ for $B'=\out(B)$. Let us assume
that $(p_1,\l_1,B_1) \lact{a} (p_2,\l_2,B_2)$. We show that
$(p_1,\l_1,\out(B_1)) \lpact{a} (p_2,\l_2,\out(B_2))$. If $a$ is
either $\tau$, $\ar(x,v)$ or $\ao(x,v)$, then the transition changes only the
state. If $a$ is $\spawn(p')$ then $B_2=B_1 \cup
\set{\spawn(p')}$ and also $\out(B_2)=\out(B_1) \cup
\set{\spawn(p')}$. If $a=\abw(x,v)$, then $B_2=B_1 \cup B
\set{\ao(x,v)}$ for some $B \incl B_1$ with $B \set{\ao(x,v)} \incl
L$. Let $B'=\set{\spawn(p') : \spawn(p')\b \in B, \text{ for some }
  \b}$. We have $B' \incl \out(B_1)$  and
$\out(B_2)=\out(B_1) \cup B' \set{\ao(x,v)}$, because $B_1$ is prefix closed.

\medskip

For the right-to-left direction we show by induction on $|\a'|$ that
$(p,\l_\init,\es)\lpact{\a'}(q,\l',B')$ for some $\l'$ and $B'\incl L'$, then
$(p,\l_\init,\es)\lact{\a}(q,\l,B)$ for some $\l$ and $B\incl L$, such that
$B' \incl \out(B)$, and $\a'$ with $\ext(\a')=\ext(\a)$.

Assume $(p,\l_\init,\es)=(p_1,\l_1,B'_1) \lpact{a_1} (p_2,\l'_2,B'_2)
\lpact{a_2} \cdots \lpact{a_{n-1}} (p_n,\l'_n,B'_n)=(q,\l',B')$. We
look for some $\l_i$ and $B_i \incl L$ such that $B'_i \incl
\out(B_i)$ and $(p_i,\l_i,B_i) \lact{\d_i a_i}
(p_{i+1},\l_{i+1},B_{i+1})$, for every $i$, for some sequences of
internal actions $\d_i$. The case where $a_i$ is $\tau$ or $\ao(x,v)$
is immediate, since only the state changes. Same holds for
$a_i=\spawn(p')$, since then $B'_{i+1} =B'_i \cup \set{\spawn(p')}$ and
$B_{i+1} =B_i \cup \set{\spawn(p')}$. In all three cases $\d_i$ is empty.

 Let us assume that
$a_i=\ar(x,v)$, so $\l'_i(x)=v$. If $\l_i(x)=v$ then we can do $a_i$
immediately, and $\d_i$ is empty. If not, since the initial value
$v_\init$ can be neither read nor written, there must be some $j<i$
such that $a_j=\abw(x,v)$. So $B'_i$ contains $\spawn(p')$ and
$\spawn(p') \ao(x,v)$, for some $p'$. By inductive assumption we have
$B'_i \incl \out(B_i)$, thus there is some sequence $\spawn(p')
\b\ao(x,v) \in B_i$, for some $\b\in\Sext^*$. So we can ``replay''
this output and do $(p_i,\l_i,B_i) \lact{\abw(x,v)}
(p_i,\l_{i+1},B_{i+1}) \lact{\ar(x,v)} (p_{i+1},\l_{i+1},B_{i+1})$
(here we have $B_i=B_{i+1}$).

The last case is when $a_i=\abw(x,v)$, so $B'_{i+1}=B'_i \cup B'
\set{\ao(x,v)}$, with $B'\incl B'_i$, $B' \set{\ao(x,v)} \incl L'$. In
particular, $B' \incl \Ssp$. Since $B' \incl \out(B_i)$ we also have
$B' \incl B_i$. Since $B' \set{\ao(x,v)} \incl \out(L)$ we can choose
some subset $B \subseteq L$ with spawn's from $B'$, and such that
$B
\set{\ao(x,v)} \incl L$. So every sequence in $B \set{\ao(x,v)}$ is of
the form $\spawn(q_j) \b_j \ao(x,v)$, with 
$\b_j$ consisting only of outputs $\ao(x',v')$, $j=1,\ldots,k$. In
addition, $B'=\set{\spawn(q_j) : j}$. Let
$\d^j$ be obtained from $\b_j$ by renaming $\ao(x',v')$ into
$\abw(x',v')$, and $\d_i=\d^1 \cdots \d^k$. We have $(p_i,\l_i,B_i)
\lact{\d_i} (p_i,\hat{\l_i},\hat{B_i}) \lact{\abw(x,v)}
(p_{i+1},\l_{i+1},B_{i+1})$, where $B'
\set{\ao(x,v)} \incl \out(B_{i+1})$ (here we have  $p_i=p_{i+1}$).
\qed
\end{proof}

Together with Lemma~\ref{lem:sig-equal} we obtain:

\begin{corollary}\label{cor:simple-out}
  Let $L =\Ext_k$ and $L'=\out(L)$. Then
  $(p,\l_\init,\es)\lact{\a}(q,\l,B)$ for some $\l$ and $B\incl L$ iff
  $(p,\l_\init,\es)\lpact{\a'}(q,\l',B')$ for some $\l'$ and $B'\incl
  L'$, with $\ext(\a)=\ext(\a')$.
\end{corollary}

The next lemma computes the set of outputs at level 0:

\begin{lemma}\label{lem:one-step-out}
  Let $\Ss$ be a pushdown dynamic parametric process with simple futures.  For any
  state $p$ of $\Ss$, and any prefix-closed set $L \subseteq \Ssp \cup
  \Ssp \Sext$ of
  outputs, consider the language
  \begin{equation*}
    K_p=\set{\ext(\alpha) :
    (p,\l_\init,\es)\lact{\a}(q,\l,B) \text{ for some } B}\ .     
  \end{equation*}
  If the number of variables is fixed, then we can determine whether
  $\ao(x,v) \in \out(K_p)$ in \NP\ (resp., \PTIME\ if $|V|$ is fixed).
\end{lemma}

\begin{proof}
  From the pushdown system of $\Ss$ we can construct a pushdown system
  $\Pp^p$ for $K_p$, by adding the $\l$ and $B$ component to the
  finite control. Doing this yields a pushdown automaton of
  exponential size, because there are exponentially many $B$. However,
  since the $B$ component grows monotonically, we can guess beforehand
  the polynomially many changes and check reachability on the new
  pushdown of polynomial size. 
\qed
\end{proof}

As a consequence we obtain with similar arguments as for signatures:

\begin{lemma}\label{lem:outk} 
  Let $\Ss$ be a pushdown dynamic parametric process with simple  futures and
  a fixed number of variables.
  Then for every $k\geq 0$, we can check membership in 
 $\out(\Ext_k)$ in \NP\ (resp., \PTIME\ if $|V|$ is fixed).
\end{lemma}

Finally we obtain:

\begin{theorem}\label{th:out}
  The reachability problem for pushdown dynamic parametric processes with
  simple  futures  and a fixed number of variables is 
  \NP-complete (resp., \PTIME\ if $|V|$ is fixed). 
\end{theorem}

\begin{proof}
  We only sketch the \NP\ lower bound, that already holds for finite-state
  systems, by a reduction from SAT. Let $C_1 \wedge \cdots \wedge C_m$ be a CNF
  formula with clauses $C_j$ over variables $x_1,\ldots,x_n$. The root
  process first spawns sub-processes with initial state $1$, then $2$, and so on
  up to $n$. Then it
  guesses a valuation $(b_1,\ldots,b_n) \in\set{0,1}^n$ by writing into
  the (unique) register $(1,b_1),\ldots,(n,b_n)$. A sub-process with
  initial state $i$ will read the value $(i,b_i)$ and then write into
  the register $C_{j_1},\ldots,C_{j_k}$, where $j_1 < \cdots < j_k$
  are the indices of those clauses that become true if $x_i$ is set to
  $b_i$. The root process needs to read $C_1,\ldots,C_n$ from the
  register in order to output $\#$.

\end{proof}


\section{Dynamic parametric processes without local variables}\label{s:nolocals}

In this section we consider another restriction of  dynamic parametric
pushdown processes, that is incomparable to the previous ones: we
allow only communication over global variables.
We show that in this case the power of the \textsf{spawn} operation is quite
limited, and that the hierarchical structure can be \emph{flattened}. 
More precisely, we give a reduction to the reachability problem of
parametric pushdown processes without local variables, where
additionally only the root  can do the \textsf{spawn} and
and will spawn just once, as its first action. 
A system with this property can be simulated by a 
\CDsystem\ as in~\cite{DBLP:conf/fsttcs/Hague11}.
We would actually need a slight extension of \CDsystems\ since the
literature considers only the variant with a single variable. 
Yet the same methods give decision algorithms for any fixed finite number of variables.
In consequence, this reduction implies that the reachability problem when each sub-process is 
realized as a pushdown system, is \PSPACE-complete.

Given a dynamic parametric process
$\Ss=(Q,G,\es,V,q_\init,v_\init,\Delta)$ without local variables
we construct a new process $\platt(\Ss)$ with one more
global variable $g_\ssp$, that we call \emph{spawn variable}.
This variable will store \textsf{spawn} demands, that is the state $p$ when
some process has performed $\spawn(p)$:
\begin{equation*}
V_{\ssp}=\set{p\in Q :
\, \spawn(p) \text{ occurs in } \D}\,.
\end{equation*}
Recall that $V_\ssp$ is finite 
since the number of \textsf{spawn} transitions in $\D$ is finite. 
The main property of $\platt(\Ss)$ will be its restricted use of $\spawn$.
 
The process $\platt(\Ss)$ is given by
$(Q',G\cup\set{g_\ssp},\es,V\cup V_\ssp,q'_\init,v_\init,\Delta')$ where 
$Q'=Q\cup\set{q'_\init,q''_\init}$
and
$\Delta'$ is obtained from $\D$ as follows.
We replace every \textsf{spawn} transition $q
\act{\spawn(p)} q'$ in $\D$ by a write $q \act{\aw(g_\ssp,p)} q'$ into the
spawn variable. In addition, from $q'_\init$ we add to $\D'$ the transition
$q'_\init \act{\spawn(q''_\init)} q_\init$, and from $q''_\init$
we add $q''_\init \act{\ar(g_\ssp,p)} p$ for every state $p\in Q$.
The result is that in $\platt(\Ss)$ the only \textsf{spawn}
operation happens from the new initial state $q'_\init$, this operation creates a
number of sub-processes all starting in state $q''_\init$. 
Subsequently, the sub-process started in $q'_\init$ goes to the state
$q_\init$ and proceeds as in $\Ss$, but instead of doing
$\spawn(p)$ it just writes $p$ into the spawn variable $g_\ssp$.
The sub-processes  that have been spawned are all in state $q''_\init$ from
which they may proceed by reading $p$ from the spawn variable.
As we can see, every reachable configuration of $\platt(\Ss)$ is a
tree of height 1,
consisting only of a root and immediate children of the root created
by the single \textsf{spawn} operation.
In order to proceed, the children of the root need to
read some $p$ value from the spawn variable.



\begin{lemma}
  A dynamic parametric process $\Ss$ without local variables  has a consistent run
  containing some external write of $\#$ if and
  only if  $\platt(\Ss)$ has one.
\end{lemma}

\begin{proof}
Since we consider processes without local variables the valuation of the
local variables is the empty function. 
In general, a configuration is a triple $(q,\l,S)$, but now we can omit $\l$. 
So in this proof s-trees will have the form
$(q,S)$, where $S$ is a finite set of s-trees. 
The configurations of $\platt(\Ss)$ have one more property as only 
the root  can do a \textsf{spawn}: the occurring s-trees all are of the form $(q,S)$
where $S$ is a set of pairs $(q',\es)$ for some state $q'$ --- implying that $S$ is
essentially a set of states.

Assume that the process $\Ss$ has a consistent run in the set semantics
containing some external write of $\#$.
To simulate this run, let $\platt(\Ss)$ start with a transition
$q'_\init\act{\spawn(q''_\init)} q_\init$ spawning sub-processes in state
$q''_\init$.
A \textsf{spawn} operation $\spawn(p)$ by some sub-process in $\Ss$ is replaced in
$\platt(\Ss)$ by a write $\aw(g_\ssp,p)$ or $\abw(g_\ssp,p)$ into the spawn variable.
This allows a child sub-process to wake up by moving from state $q''_\init$
to $p$.
So when simulating the run of $\Ss$, we keep the invariant that the
root processes in $\Ss$ and $\platt(\Ss)$ are in the same state, and
for every state $q$ but $q''_\init$:  state $q$ appears in a
configuration of $\Ss$ iff it appears in the corresponding
configuration of $\platt(\Ss)$. 

Conversely, consider a consistent run of $\platt(\Ss)$ reaching an
external write of $\#$. 
It must start with
$(q'_\init,\es)\sact{\spawn(q'_\init)}(q_\init,\set{q''_\init})$.
We show that for every consistent run $(q_\init,\set{q''_\init})\sact{\a'}
(q,S')$ of $\platt(\Ss)$ we can find a consistent sequence $\a$ such that
$(q_\init,\es)\sact{\a'} (q,S)$ where
$S$ consists of states from $S'$ apart $q''_\init$ plus additionally
the state $p$ from the last write $\aw(g_\ssp,p)$ to the spawn
variable in $\a'$; if there is such a write.

The proof of this claim is by induction on the length of $\a'$
considering possible transitions from $(q,S')$ one by one:

\begin{itemize}
\item Root process read or write: if $(q,S') \sact{\aw(g,v)}
  (q_1,S')$ with $g\not=g_\ssp$ then $(q,S)\sact{\aw(g,v)}
  (q_1,S)$. Same for $\ar(g,v)$.
\item Root process \textsf{spawn}: if $(q,S') \sact{\aw(g_\ssp,p)}
  (q_1,S')$ then $(q,S)\sact{\spawn(p)} (q_1,S \cup\set{(p,\emptyset)})$.
\item Sub-process read or write: let $(q,S') \sact{\abw(g,v)}
  (q,S'\cup\set{(q_2,\es)})$ for $g\not=g_\ssp$, and $(q_1,\es)
  \sact{\aw(g,v)} (q_2,\es)$ for some  $(q_1,\es) \in S'$.
  Since $q_1$ occurs in $S$ by our induction hypothesis we get $(q,S)
  \sact{\abw(g,v)}  (q,S_2)$ for some suitable $S_2$.
  Same for $\abr(g,v)$.
\item Sub-process \textsf{spawn}: let $(q,S') \sact{\aw(g_\ssp,p)}
  (q,S'\cup \set{(q_2,\es)})$ for some $q_1 \in S'$ with $(q_1,\es)
  \sact{\aw(g_\ssp,p)} (q_2,\es)$. Since $q_1$ occurs in $S$, by induction
  we can pick an s-tree $(q_1,S_1)$ within $S$ and add a sibling
  $(q_1,S_1)\sact{\spawn((p,\es))}(q_2,S_1 \cup \set{(p,\emptyset)})$. This gives us 
  $(q,S)\sact{\t}(q,S')$ with $S'$ containing $(q_2,S_1 \cup \set{(p,\es)})$.

\item Sub-process ``wake-up'': let $(q,S') \sact{\ar(g_\ssp,p))}
(q,S'\cup\set{p})$. This transition is not
simulated by any transition in $\Ss$ (note that $p$ already occurs in
$S$ by our induction hypothesis and the fact that $\a$ is consistent).
\end{itemize}
\qed
\end{proof}

\begin{theorem}
  Let the number of variables be fixed.
  The reachability problem for pushdown dynamic parametric processes
  without local variables is \PSPACE-complete.
\end{theorem}
\begin{proof}
  If $\Ss$ is given by a pushdown automaton, then so is
  $\platt(\Ss)$. 
  Now $\platt(\Ss)$ has only one $\spawn$  at the very beginning of
  its execution. 
  Such a system can be simulated by a
  \CDsystem~\cite{DBLP:conf/fsttcs/Hague11,DBLP:conf/cav/Durand-Gasselin15,DBLP:journals/jacm/EsparzaGM16} of
  the same size: the leader system $D$ is $\platt(\Ss)$ with the initial
  state $q_\init$, and the contributor system $C$ is $\platt(\Ss)$
  with the initial state $q''_\init$. 
  A small obstacle is that in the literature \CDsystems\ were defined
  with only one variable, while we need at least two.
  It can be checked that the reachability problem for \CDsystems\
  given by pushdown automata is
  \PSPACE-complete~\cite{DBLP:conf/cav/Durand-Gasselin15,DBLP:journals/jacm/EsparzaGM16}
  even for systems with more than one variable as long as the number
  of variables is fixed.
  Our reduction then shows that our reachability problem is in \PSPACE.
  Since \CDsystems\ can be directly simulated by our model we also get  the
  \PSPACE\ lower bound.
\qed
\end{proof}





\section{Conclusions}\label{s:conclusion}

We have studied systems with parametric process creation where sub-processes may communicate
via both global and local shared variables.
We have shown that under mild conditions, reachability 
for this model is decidable. 
The algorithm relies on the abstraction of the behavior of
the created child sub-processes by means of finitely many minimal behaviors. 
This set of minimal behaviors is obtained by a fixpoint computation
whose termination relies on well-quasi-orderings.
\igw{changed}This bottom-up approach is different from the ones used
before~\cite{DBLP:conf/fsttcs/Hague11,DBLP:journals/jacm/EsparzaGM16,DBLP:conf/concur/TorreMW15}. 
In particular, it avoids computing a downward closure, thus showing that
computability of the downward 
closure is not needed in the general decidability results
on flat systems from~\cite{DBLP:conf/concur/TorreMW15}.

We have also considered special cases for pushdown dynamic parametric
processes where we 
obtained solutions of a relatively low complexity.   
In absence of local variables we have shown that reachability 
can be reduced to reachability for systems without dynamic sub-process
creation,
implying that reachability is \PSPACE-complete.
For the (incomparable) case where communication is restricted to child
sub-processes reporting their results to siblings and
their parents, we have also provided a dedicated method with \DEXPTIME\
complexity. 
We conjecture that this bound is tight. 
\igw{added}Finally, when sub-processes can report results only to their parents,
the problem becomes just \NP-complete.

An interesting problem for further research is to study
the reachability of a particular \emph{set} of configurations 
as considered, e.g., for dynamic pushdown networks
\cite{BouajjaniMOT05}.
One such set could, e.g., specify that all children of a given sub-process have terminated.
For dynamic pushdown networks with nested or contextual locking, such kinds of barriers
have been considered in \cite{DBLP:conf/vmcai/GawlitzaLMSW11,DBLP:conf/sas/LammichMSW13}.
It remains as an intriguing question whether or not similar concepts
can be handled also for 
dynamic parametric processes.

\bibliographystyle{abbrv}
\bibliography{../bib}

\end{document}